\documentclass[article, pra, showpacs, twocolumn, showkeys, secnumarabic, aps, amsmath, amssymb, nofootinbib, superscriptaddress, nolongbibliography, floatfix, table-of-contents, dblfloatfix]{revtex4-2}
\usepackage[utf8]{inputenc}
\usepackage{soul, amsthm,xcolor}
\usepackage{caption,comment}
\usepackage{subcaption}
\usepackage{float}
\usepackage{graphicx}
\usepackage{mathrsfs,physics}
\usepackage[colorlinks, breaklinks, urlcolor={cyan}, linkcolor={red}, citecolor={blue}]{hyperref}
\usepackage{array}
\usepackage{amsmath,amsthm}
\usepackage{amssymb}
\usepackage{type1cm}
\usepackage[export]{adjustbox}
\usepackage{dsfont}
\usepackage{lettrine}
\usepackage[english]{babel}
\usepackage{lmodern}
\usepackage{microtype}
\usepackage{booktabs}
\usepackage[T1]{fontenc}
\usepackage[boxed, vlined]{algorithm2e}
\usepackage{braket}
\usepackage{xcolor}
\usepackage{orcidlink}
\usepackage{braket}
\usepackage{bm}
\usepackage{bbold}
\usepackage{tikz}

\usepackage{upgreek}
\usepackage{xcolor}

\newtheorem{Theorem}{Theorem}
\newtheorem{Lemma}{Lemma}

\newtheorem{Remark}{Remark}

\newtheorem{Proposition}{Proposition}

\usepackage{bbm}

\begin{document}
\title{Convexity and non-Markovianity of Weyl Maps}
\author{Wen Xu \orcidlink{0000-0002-0072-2949} }
\email[]{xuwen_quantum@ecut.edu.cn}
\affiliation{School of Science, East China University of Technology, Nanchang 330013, China}

\author{Vinayak Jagadish \orcidlink{0000-0002-9637-5194}}
\email[]{vinayakj@am.amrita.edu}
\affiliation{Department of Computer Science and Engineering, Amrita School of Computing, Amrita Vishwa Vidyapeetham, Amritapuri 690525, India}

\begin{abstract}
We investigate the emergence of non-Markovian dynamics in finite-dimensional open quantum systems governed by Weyl dynamical maps and their convex combinations. Using the Hermite normal form, we provide a complete classification of the subgroups of the discrete phase space $\mathbb{Z}_d \times \mathbb{Z}_d$, establishing the algebraic framework underlying the Weyl maps. We characterize isotropic Weyl dynamical maps that generate Markovian semigroups and show that anisotropic Weyl maps with nonuniform weight distributions cannot possess the semigroup property. Furthermore, we analyze the role of convexity in the generation and suppression of memory effects. Remarkably, we prove that convex combinations of eternally non-Markovian Weyl dephasing maps can generate Markovian semigroups, demonstrating that non-Markovianity is not additive under mixing. Conversely, we establish a general condition under which convex mixtures of $N$ distinct Weyl semigroups exhibit eternal non-Markovianity. In contrast to the qubit Pauli setting, we further identify the existence of irreducible eternally non-Markovian Weyl dephasing maps, namely, individual dynamical maps that display eternal memory effects without requiring any mixing mechanism. Finally, explicit qutrit examples illustrate the transition among Markovian, non-Markovian and eternally non-Markovian regimes. Our results uncover a fundamental connection among finite phase-space algebra, convex structures, and quantum memory effects, thereby extending the theory of non-Markovian dynamics beyond the Pauli framework.
\end{abstract}

\maketitle
\section{Introduction}
The dynamics of quantum systems interacting with surrounding environments~\cite{haroche_exploring_2006,Quanta77} lies at the heart of quantum information science and quantum technologies. Such interactions inevitably lead to decoherence and dissipation, which lead to decoherence and dissipation, thereby degrading quantum coherence and limiting the performance of quantum devices. The theory of open quantum systems provides the fundamental framework for modeling these processes and for developing strategies to mitigate their detrimental effects.

Early theoretical treatments of open quantum systems frequently relied on the Markovian approximation, in which environmental correlations decay rapidly and the system dynamics becomes effectively memoryless. Under this assumption, the evolution is generated by quantum dynamical semigroups and admits the well-known Gorini--Kossakowski--Sudarshan--Lindblad (GKSL) structure \cite{lindblad1976, gorini1976}. However, extensive theoretical and experimental studies have demonstrated that many realistic physical systems exhibit significant memory effects, giving rise to non-Markovian dynamics \cite{rivasreview, breuer_colloquium:_2016, de_vega_dynamics_2017, li_concepts_2017, chruscinski2022}. Understanding and characterizing such non-Markovian behavior has become an important problem in quantum information theory and quantum control.

Several complementary approaches have been developed to quantify and interpret quantum non-Markovianity. Prominent among these are the CP-divisibility framework introduced by Rivas, Huelga, and Plenio \cite{rivas_entanglement_2010, hall2010}, and the information-theoretic approach based on distinguishability backflow proposed by Breuer, Laine, and Piilo \cite{breuer_measure_2009}. Although these perspectives emphasize different operational aspects, they capture the central idea that non-Markovian dynamics arises from the presence of memory effects in the system–environment interaction.

An interesting and somewhat counterintuitive mechanism for generating non-Markovian behavior arises through \emph{convex combinations} of quantum dynamical maps. Several studies have shown that mixing Markovian evolutions can produce non-Markovian dynamics according to standard criteria \cite{wolf2008, megier_eternal_2017, jagadish_convex_2020, jagadish_nonqds_2020, wudarski2015, wudarski2016markovian, siudzinskajpa2020}. In particular, geometric analyses of Pauli dynamical maps have revealed that the set of CP-divisible maps is not convex, so that convex mixtures of Markovian channels may lie outside the Markovian region \cite{jagadish_convex_2020, jagadish_nonqds_2020}. Convex mixing has also been shown to generate or conceal noninvertibility in open-system dynamics \cite{siudzinska_markovian_2021, jagadish2022noninvertibility, jagadish_measureinvert_2022, jagadish_nonivertible_2023, siudzinska_jpa_2022, xu2024markovian}. Remarkably, the generation of non-Markovian dynamics through convex mixtures of Pauli semigroups has recently been experimentally  demonstrated on an NMR quantum processor \cite{gulati2024}. These results highlight a subtle interplay between the convex structure, divisibility properties, and memory effects in quantum dynamics.

Most existing studies on convex structures and non-Markovianity have focused on qubit systems described by Pauli dynamical maps. However, many emerging quantum platforms naturally involve higher-dimensional Hilbert spaces, including photonic, atomic, and superconducting architectures employing qudit encodings. For such systems, the natural generalization of Pauli operators is provided by discrete Weyl operators, which form a unitary operator basis on $\mathbb{C}^d$ and are closely related to the Weyl--Heisenberg group \cite{Weyl1927,weyl1950theory,Schwinger1960}. Quantum dynamical maps constructed from Weyl operators provide a natural framework for describing noise and decoherence in higher-dimensional quantum systems.

Despite their importance, the structural properties of Weyl dynamical maps have received considerably less attention than their qubit counterparts. In particular, the convex structure of Weyl semigroups, their divisibility properties, and the mechanisms by which convex mixing can generate or suppress non-Markovian behavior in higher dimensions remain largely unexplored. Addressing these questions is essential for developing a deeper understanding of memory effects in high-dimensional open quantum systems.

In this work, we extend several foundational results known for Pauli dynamical maps to the broader setting of qudit Weyl dynamical maps. By exploiting the algebraic structure of the discrete phase space $\mathbb{Z}_d \times \mathbb{Z}_d$, we analyze the property of semigroup, the divisibility conditions, and the convex mixtures of Weyl dynamical maps. Our results reveal new features that arise specifically in higher dimensions and demonstrate how convex combinations of Weyl semigroups can both generate and suppress non-Markovian behavior depending on the underlying subgroup structure.

The paper is organized as follows. In Sec.~\ref{sec:2}, we introduce Weyl dynamical maps together with the concepts of quantum dynamical semigroups, CP-divisibility, and eternal non-Markovianity. We also analyze the algebraic structure of Weyl dynamical maps through a complete classification of the subgroups of $\mathbb{Z}_d \times \mathbb{Z}_d$. In Sec.~\ref{sec:3}, we derive the necessary and sufficient conditions under which isotropic Weyl dynamical maps form Markovian semigroups. In Sec.~\ref{sec:4}, we show that convex combinations of eternally non-Markovian Weyl dephasing maps can generate Markovian semigroups. Furthermore, we establish a general condition under which convex mixtures of $N$ distinct Weyl semigroups lead to eternal non-Markovianity. In Sec.~\ref{sec:5}, we illustrate these phenomena with explicit examples for the case $d=3$. Finally, in Sec.~\ref{sec:6}, we summarize our main results and discuss possible directions for future work.

\section{Preliminaries}\label{sec:2}
\subsection{Weyl dynamical maps} \label{subsec:II A}
Weyl operators, also known as discrete Weyl-Heisenberg operators \cite{Weyl1927,weyl1950theory,Schwinger1960}, are unitary matrices that generalize the Pauli matrices to higher-dimensional Hilbert spaces and exhibit rich algebraic structures. Given a $d$-dimensional Hilbert space $\mathbb{C}^d$, the Weyl operators $\{U_{kl}\}$ are defined for $k,l \in \mathbb{Z}_d$ as:

\begin{equation*}
    U_{kl} = \sum_{m=0}^{d-1} \omega^{km} \ket{m}\bra{m + l},
    \label{eq:weyl_def}
\end{equation*}
where $\omega = e^{2\pi \mathrm{i} / d}$ is a primitive $d$-th root of unity, and addition in the ket indices is performed modulo $d$. These operators form a basis for the space of linear operators in $\mathbb{C}^d$ and exhibit a rich algebraic structure with the following key properties:
   \begin{eqnarray*}
        U_{kl} U_{rs} &=&\omega^{lr - ks} U_{rs} U_{kl},\nonumber\\
 U_{kl}^\dagger &=& U_{-k, -l},\nonumber\\
 \text{Tr}(U_{kl}^\dagger U_{rs}) &=& d \, \delta_{kr} \delta_{ls},\nonumber\\
 U_{k + d, l} &=& U_{k, l + d} =U_{kl}.
    \end{eqnarray*}
The Weyl dynamical map is defined by 
\begin{equation}\label{double-indices-map}
\mathcal{E}(t)(\rho)=\sum_{(i,j)\in \mathbb{Z}_d\times\mathbb{Z}_d}\, p_{ij}(t)\, U_{ij}\rho U_{ij}^\dagger.
\end{equation}
Note that the indices $(i,j)$ take values in the direct product of two cyclic groups $\mathbb{Z}_d\times\mathbb{Z}_d$ which serves as a discrete phase space label for the Weyl operators. While this set carries a natural abelian group structure under addition modulo $d$, the Weyl operators provide only a projective unitary representation of $\mathbb{Z}_d \times \mathbb{Z}_d$.

The Weyl operators are eigenoperators of the dynamical map:
\begin{equation*}\label{eig-eq}
\mathcal{E}(t)(U_{kl})=\lambda_{kl}(t)U_{kl},
\end{equation*}
where the eigenvalues $\lambda_{kl}(t)$ are given by
\begin{equation}\label{eig-double-index}
\lambda_{kl}(t)=\sum\limits_{(i,j)\in \mathbb{Z}_d\times\mathbb{Z}_d}\omega^{jk-il}\; p_{ij}(t).
\end{equation}

  \subsection{Quantum dynamical semigroups, CP-divisibility and eternal non-Markovianity}
\label{subsec:CPdiv}

Before analyzing the algebraic structure of Weyl maps, we briefly recall the criteria used to characterize memory effects in the dynamics of open quantum systems. We consider a family of dynamical maps $\{\mathcal{E}(t)\}_{t\geq0}$ that describe the reduced evolution of the system,
\begin{equation*}
    \rho(t) = \mathcal{E}(t)[\rho(0)],
    \label{eq:dynamical_map}
\end{equation*}
where each map $\mathcal{E}(t)$ is completely positive and trace preserving (CPTP), and $\mathcal{E}(0)=\mathbb{1}_d$ is the identity operator.

The evolution is governed by the time-local master equation
\begin{equation*}
    \frac{\mathrm{d}}{\mathrm{d}t}\rho(t) = \mathcal{L}(t)\rho(t),
    \label{eq:time_local_ME}
\end{equation*}
where the generator $\mathcal{L}(t)$ admits the GKLS form
\begin{align}
    \mathcal{L}(t)(\rho)
    =-\mathrm{i}[H(t),\rho]+ 
    \sum_{\alpha}\gamma_{\alpha}(t)
    \left(
    L_{\alpha}\rho L_{\alpha}^{\dagger}
    - \frac{1}{2}\{L_{\alpha}^{\dagger}L_{\alpha},\rho\}
    \right).
    \label{eq:GKLS_time_dep}
\end{align}
Here, $H(t)$ is an effective Hamiltonian (possibly time-dependent), $\{L_{\alpha}\}$ are Lindblad operators, and the real-valued functions $\gamma_{\alpha}(t)$ are time-dependent decay rates.

For a Weyl dynamical map, the generator $\mathcal{L}(t)$ takes a particularly simple form. Identifying the Lindblad operators $\{L_\alpha\}$ with the unitary Weyl operators $\{U_{ij}\}$ and assuming a vanishing effective Hamiltonian $H(t)=0$, one can show that Eq. (\ref{eq:GKLS_time_dep}) reduces to \cite{xu2024markovian} 
\begin{equation}\label{generator-expression}
\mathcal{L}(t)(\rho)=\sum\limits_{(i,j)\in\mathbb{Z}_d\times\mathbb{Z}_d\setminus\{(0,0)\}}\gamma_{ij}(t)(U_{ij}\rho U_{ij}^{\dagger}-\rho),
\end{equation}
where the decay rates $\gamma_{ij}(t)$ are
\begin{equation}\label{decay-rate-expression}
\gamma_{ij}(t)=\frac{1}{d^{2}}\sum\limits_{(k,l)\in\mathbb{Z}_{d}\times\mathbb{Z}_{d}}\omega^{il-jk}~\mu_{kl}(t).
\end{equation}
Here, the eigenvalues of the generator are $\mu_{kl}(t)=\frac{\dot{\lambda}_{kl}(t)}{\lambda_{kl}(t)}$ due to $\mathcal{L}(t)=\dot{\mathcal{E}}(t) \circ \mathcal{E}^{-1}(t)$ if the map $\mathcal{E}(t)$ is invertible.

Throughout this work, we characterize memory effects according to the Rivas--Huelga--Plenio (RHP) criterion~\cite{rivas_entanglement_2010}, which associates Markovianity with the complete positivity of intermediate dynamical maps, or equivalently, with the CP-divisibility of the evolution.

\begin{enumerate}
    \item \textbf{CP-divisibility (Markovianity).}  
    A dynamical map $\{\mathcal{E}(t)\}_{t\geq0}$ is said to be \emph{CP-divisible} if, for any $t \ge s \ge 0$, there exists a CPTP map $\mathcal{E}(t,s)$ such that
    \begin{equation*}
        \mathcal{E}(t) = \mathcal{E}(t,s)\mathcal{E}(s).
        \label{eq:CP_divisibility}
    \end{equation*}
    For an invertible time-local master equation of Eq. \eqref{eq:GKLS_time_dep}, CP-divisibility is equivalent to the condition
    \begin{equation*}
        \gamma_{\alpha}(t) \geq 0
        \qquad
        \forall\, \alpha,\;\forall\, t \ge 0.
        \label{eq:positive_rates}
    \end{equation*}
    Dynamical evolution satisfying this condition is termed \emph{Markovian} according to the RHP criterion.

    \item \textbf{Markovian semigroup.}  
    If the decay rates are time-independent nonnegative constants, i.e.,
    \begin{equation*}
        \gamma_{\alpha}(t) \equiv \gamma_{\alpha} \ge 0,
        \label{eq:constant_rates}
    \end{equation*}
    then the generator $\mathcal{L}(t)$ becomes time-independent and the corresponding dynamical map forms a one-parameter quantum dynamical semigroup,
    \begin{equation*}
        \mathcal{E}(t+s) = \mathcal{E}(t)\mathcal{E}(s),
        \qquad t\ge s \ge 0.
        \label{eq:semigroup_property}
    \end{equation*}
    This represents the strongest notion of Markovianity and corresponds to strictly memoryless evolution.

    \item \textbf{Non-Markovianity.}  
    The dynamics is classified as \emph{non-Markovian} whenever CP-divisibility is violated. For time-local master equation, this occurs if there exists at least one decay channel $\alpha$ and a time $t>0$ such that
    \begin{equation*}
        \gamma_{\alpha}(t) < 0.
        \label{eq:negative_rate}
    \end{equation*}
    Negative decay rates indicate that the failure of the generator to be of GKLS form at that time and imply that the intermediate map $\mathcal{E}(t,s)$ is no longer completely positive.

    \item \textbf{Eternal non-Markovianity (ENM).}  
Eternal non-Markovianity denotes a particularly strong violation of CP-divisibility. A dynamics is said to be eternally non-Markovian if, for at least one decay channel $\alpha$ such that
\begin{equation*}
    \gamma_{\alpha}(0)=0,
    \qquad
    \gamma_{\alpha}(t) < 0 \;\; \forall\, t>0^{+}.
    \label{eq:ENM_condition}
\end{equation*}
In this regime, the CP-divisibility is violated immediately after the initial time, and the evolution is non-Markovian for all $t>0$, even though each map $\mathcal{E}(t)$ remains CPTP for all $t\ge0$.

\end{enumerate}

\subsection{The algebraic structure of Weyl dynamical maps}
\subsubsection{From weights to subgroups}
The dynamical properties of a Weyl map $\mathcal{E}(t)$ are completely determined by the time-dependent coefficients (or weights) $p_{ij}(t)$ appearing in \eqref{double-indices-map}. In particular, the algebraic structure of the dynamics is closely related to the support of the probability distribution 
$
\mathrm{supp}(p)
=
\{(i,j)\in\mathbb{Z}_d\times\mathbb{Z}_d : p_{ij}(t)\neq0\}.$
In this work, we restrict our attention to the indices belonging to the support for which the participating Weyl operators constitute a subgroup $G\subseteq \mathbb{Z}_d\times\mathbb{Z}_d.$ A systematic analysis of such dynamics requires a classification of the subgroups of the discrete phase space. Since a subgroup may admit many different generating sets, it is convenient to adopt a canonical description based on lattice theory. Viewing subgroups of $\mathbb{Z}_d\times\mathbb{Z}_d$ as integer lattices modulo $d$, we employ the Hermite normal form (HNF)~\cite{newman1972}, which provides a unique matrix representation for each subgroup. This approach provides a complete and non-redundant classification of all subgroups of $\mathbb{Z}_d\times\mathbb{Z}_d$.

\subsubsection{Classification via Hermite normal form}
\begin{Lemma}[HNF classification]\label{lem:HNF}
Every subgroup $G\subseteq \mathbb{Z}_d\times\mathbb{Z}_d$ admits a unique representation of the Hermite normal form, which is generated by the rows of an upper triangular matrix
\begin{equation}\label{eq:HNF}
M=\begin{pmatrix}
m & w \\
0 & n
\end{pmatrix}.
\end{equation}
Here, the parameters satisfy
\begin{enumerate}
    \item [\rm (1)] $m \mid d$ and $n \mid d$,
    \item [\rm (2)] $0 \leq w < n$,
    \item [\rm (3)] $n \mid \frac{wd}{m}$.
\end{enumerate}
\end{Lemma}
The proof of Lemma \ref{lem:HNF} is given in Appendix \ref{sec:AppendixA}.

The subgroup generated by \eqref{eq:HNF} is
\begin{align}\label{eq:G-general}
G&=\langle(m,w),(0,n)\rangle \nonumber\\
&=\{(mu,wu+nv)\mid u\in\mathbb{Z}_{d/m},\,v\in\mathbb{Z}_{d/n}\}
\end{align}
with order
\begin{equation*}\label{eq:order-G}
|G|=\frac{d^2}{mn}.
\end{equation*}

The structural classification of the subgroup $G$ depends on whether the second generator $(0,n)$ is redundant or essential. We define the redundancy threshold 
$$\nu= \gcd(\frac{wd}{m}, d)$$  
as the greatest common divisor of $\frac{wd}{m}$ and $d.$
\begin{enumerate}
\item [\rm (1)] \textbf{Type 1 (cyclic subgroups)}\\ 
If 
\begin{equation}\label{eq:Type1-condition}
n=\nu,
\end{equation}
then
\begin{equation}\label{eq:Type1}
G=\langle(m,w)\rangle,
\qquad
|G|=\frac{d}{\gcd(m,w)}.
\end{equation}
In this case, the second generator $(0,n)$ is redundant and is generated by repeated addition of $(m,w)$.
\item [\rm (2)] \textbf{Type 2 (split rank-$2$ subgroups)} \\
If
\begin{equation}\label{eq:Type2-condition}
w=0,\qquad n<d,
\end{equation}
then
\begin{equation}\label{eq:Type2}
G=m\mathbb{Z}_d \times n\mathbb{Z}_d
=\{(mu,nv)\mid u\in\mathbb{Z}_{d/m},\,v\in\mathbb{Z}_{d/n}\},
\end{equation}
which is a direct product of two cyclic groups.

\item [\rm (3)] \textbf{Type 3 (non-split rank-$2$ subgroups)} \\
If
\begin{equation}\label{eq:Type3-condition}
w\neq 0
\quad\text{and}\quad
n<\nu,
\end{equation}
then $G$ has rank $2$ but the subgroup is not expressed in the split form $
m\mathbb{Z}_d\times n\mathbb{Z}_d$. Such subgroup is given by \eqref{eq:G-general}.
\end{enumerate}

\begin{Remark}
Lemma \ref{lem:HNF} provides a complete classification of the subgroups of $\mathbb{Z}_d\times\mathbb{Z}_d$, thereby extending the analysis presented in Ref. \cite{xu2024markovian}. As an example, consider a subgroup of type 3 in dimension $d=4$. Setting the diagonal parameters to $m=1$ and $n=2$ which requires the off-diagonal parameter to be $w=1$. This yields a subgroup of order $|G|=8$, which is generated by
\[
    G=\langle(1,1),(0,2)\rangle.
\]

\end{Remark}

\subsubsection{Symplectic structure and dual subgroups}
The algebraic properties of the Weyl operators and the dynamics are deeply rooted in the symplectic structure of the discrete phase space $\mathbb{Z}_d \times \mathbb{Z}_d$. To simplify the notation without introducing ambiguity, we denote each element $(i,j) \in \mathbb{Z}_d \times \mathbb{Z}_d$ in the discrete phase space by a single index $u$. Unless otherwise specified, this single index convention will be adopted throughout the remainder of this paper. For any two such indices $u = (i, j)$ and $v = (k, l)$, the \textit{symplectic inner product} is defined as
\begin{equation}\label{symplectic_inner_product}
    u \wedge v := jk - il \pmod d.
\end{equation}
This product naturally characterizes the commutation relations of the Weyl basis, which can be expressed in a compact form as $U_u U_v = \omega^{u \wedge v} U_v U_u$, consistent with the properties described in subsection \ref{subsec:II A}.

Given a subgroup $G \subseteq \mathbb{Z}_d \times \mathbb{Z}_d$, we define its \textit{symplectic dual subgroup}, denoted by $G^\perp$, as the set of all elements in the discrete phase space that commute with every element in $G$:
\begin{equation}\label{symplectic_dual_subgroup}
    G^\perp = \{ v \in \mathbb{Z}_d \times \mathbb{Z}_d \mid u \wedge v \equiv 0 \pmod d, \, \forall u \in G \}.
\end{equation}
This dual subgroup $G^\perp$ plays a pivotal role in the spectral analysis of Weyl dynamical maps. The action of the map on a Weyl basis operator $U_v$ yields \begin{equation}\mathcal{E}(t)(U_v) = \lambda_v(t) U_v, v\in\mathbb{Z}_d \times \mathbb{Z}_d.\end{equation} Using the symplectic inner product notation \eqref{symplectic_inner_product}, the eigenvalues $\lambda_v(t)$ of \eqref{eig-double-index} can be expressed as
\begin{align}\label{eq:eigenvalue_general}
\lambda_{v}(t) &=\sum_{u \in G} \omega^{u \wedge v}p_u(t)\nonumber\\
&=1-\sum_{u \in G\setminus\{0\}} p_u(t) + \sum_{u \in G\setminus\{0\}} \omega^{u \wedge v}p_u(t). 
\end{align}
Crucially, the evaluation of $\omega^{u \wedge v}$ depends on the relation between the index $v$ of  eigenoperator and the dual subgroup $G^\perp$ (i.e., whether $v \in G^\perp$).

Furthermore, the corresponding time-local generator acts as $\mathcal{L}(t)(U_v) = \mu_v(t) U_v$, where $\mu_v(t) = \dot{\lambda}_v(t) / \lambda_v(t)$ represent the eigenvalues of the generator. 
According to \eqref{decay-rate-expression}, the decay rates $\gamma_{\alpha}(t)$ of the map $\mathcal{E}(t)$ are given by
\begin{equation}\label{eq:gamma_def_appendix}
\gamma_{\alpha}(t)=\frac{1}{d^2} \sum_{v\in\mathbb{Z}_d\times\mathbb{Z}_d}\omega^{-\alpha \wedge v } \mu_v(t),
\end{equation}
where $\alpha \in \mathbb{Z}_d \times \mathbb{Z}_d\setminus{\{(0,0)\}}$ serves as the compact single index shorthand for the double indices. Such a symplectic structure provides a fundamental tool for our investigations of eternal non-Markovianity.

\color{black}

\subsubsection{Counting subgroups of order \texorpdfstring{$K\ge2$}{K>2} }

Let $d=\prod_{i=1}^s p_i^{e_i}$ be the prime factorization of the dimension of the system. By Lemma \ref{lem:HNF}, any subgroup $G \subseteq \mathbb{Z}_d \times \mathbb{Z}_d$ of order $K$ (where $K \mid d^2$) corresponds to a matrix with Hermite normal form such that 
\begin{equation}\label{K_order}
|G| = \frac{d^2}{mn}=K. 
\end{equation}
This implies $mn = \frac{d^2}{K}$. Let us define the set of diagonal parameters for a given order $K$ as
\begin{equation*}
\mathcal{S}_K = \left\{ (m,n) \in \mathbb{N}^2 \;\bigg|\; m \mid d, \ n \mid d, \ mn = \frac{d^2}{K} \right\}.
\end{equation*}

In the following, we employ the methods of number theory to derive the total count of subgroups for any order $K\ge2$, and we determine the specific order at which the total number of subgroups attains its maximum.

\begin{Lemma}[Counting subgroups of arbitrary order]\label{lem:counting_general}
Let $\mathcal{N}(K)$ denote the total number of subgroups of $\mathbb{Z}_d \times \mathbb{Z}_d$ of order $K$.

\begin{enumerate}
\item [\rm 1.] \textbf{Total count.} The total number of subgroups of order $K$ is given by the sum over the parameters set $\mathcal{S}_K$:
\begin{equation}\label{eq:total_K}
\mathcal{N}(K) = \sum_{(m,n) \in \mathcal{S}_K} \gcd\left(n, \frac{d}{m}\right),
\end{equation}
where $\gcd\left(n, \frac{d}{m}\right)$ is the greatest common divisor of $n$ and $\frac{d}{m}.$
\item [\rm 2.] \textbf{Maximality.} The function $\mathcal{N}(K)$ attains its maximum when
\begin{equation*}
K=d.
\end{equation*}
In this specific maximal case, the set $\mathcal{S}_d$ and the total count $\mathcal{N}(d)$ reduces to
\begin{align*}
\mathcal{S}_d&=\{(m, n) \in \mathbb{N}^2  \mid mn=d\}, \\
\mathcal{N}(d)& =\sigma_1(d) = \sum_{m \mid d} \frac{d}{m} = \prod_{i=1}^s\frac{p_i^{e_i+1}-1}{p_i-1},
\end{align*}
respectively, where $\sigma_1(d)$ is the divisor sum function.
\end{enumerate}
\end{Lemma}
The proof of Lemma \ref{lem:counting_general} is given in Appendix \ref{sec:AppendixB}.

\section{Semigroup property of isotropic Weyl dynamical maps}\label{sec:3}

Before analyzing the properties of convex combinations of Weyl dynamical maps, we first discuss the conditions for an individual Weyl dynamical map to be a semigroup.

\subsection{Absence of semigroup property in anisotropic Weyl maps}
For the three types of subgroups $G$ described above, we discuss what kind of Weyl dynamical map (\ref{double-indices-map}) cannot be a semigroup. 

\begin{Proposition}\label{prop:no_semigroup_anisotropic}
Consider the anisotropic Weyl dynamical map
\begin{equation}\label{anisotropic proportional map}
\mathcal{E}_{\mathrm{aniso}}(t)(\rho)
=
(1-p(t))\rho
+
p(t)\sum_{u\in G\setminus\{0\}}w_u U_u\rho U_u^\dagger,
\end{equation}
where the coefficients $w_u$ satisfy $\sum_{u \in G \setminus \{0\}} w_{u} = 1$ and $w_{u} \ge 0$, which are time-independent and non-uniform. If the map possesses at least two distinct nontrivial eigenvalues, then the family $\{\mathcal{E}_{\mathrm{aniso}}(t)\}_{t\ge0}$ cannot form a quantum dynamical semigroup generated by a single scalar probability function $p(t)$.
\end{Proposition}

The proof of Proposition \ref{prop:no_semigroup_anisotropic} is given in Appendix \ref{sec:Appendix:ani-pro_map}.

\subsection{Necessary and sufficient conditions for isotropic Weyl maps to form Markovian semigroups} 

By Proposition \ref{prop:no_semigroup_anisotropic}, if the weight distribution $w_u=\frac1{|G|-1}$ is uniform for all $u\in G\setminus\{0\}$, then \eqref{anisotropic proportional map} reduces to the isotropic case:
\begin{equation}\label{three-types-Weyl-maps}
\mathcal{E}_{\textnormal{iso}}(t)(\rho)=(1-p(t))\rho+\frac{p(t)}{|G|-1}\sum_{u\in G\setminus\{0\}}U_{u}\rho U^{\dagger}_{u},
\end{equation}
where $|G|\geq2.$ The isotropic Weyl dynamical maps (\ref{three-types-Weyl-maps}) correspond to the following three classes of subgroups:
\begin{enumerate}
\item [\rm (1)] The isotropic Weyl dynamical maps corresponding to the cyclic subgroups (\ref{eq:Type1}) with the condition (\ref{eq:Type1-condition}) for any $d\geq2.$

\item [\rm (2)] The isotropic Weyl dynamical maps corresponding to the subgroups (\ref{eq:Type2}) with the condition (\ref{eq:Type2-condition}) for any $d\geq2.$

\item [\rm (3)] The isotropic Weyl dynamical maps corresponding to the subgroups (\ref{eq:G-general}) with the condition (\ref{eq:Type3-condition}) for any $d\geq4.$
\end{enumerate}

\begin{Remark}
In particular, if $|G|=2$ and $u\in G\setminus \{0\}$, then \textnormal{(\ref{three-types-Weyl-maps})} reduces to the Weyl dephasing map 
\begin{equation}\label{double_Weyl_dephasing_maps}
\mathcal{E}_{u}^p(t)(\rho)=(1-p(t))\rho+p(t)U_{u}\rho U_{u}^\dagger.
\end{equation}
\end{Remark}

In what follows, we establish the necessary and sufficient conditions for the isotropic Weyl maps (\ref{three-types-Weyl-maps}) to form Markovian semigroups.

\begin{Proposition}\label{Prop1-weyl-semigroup-condition}
The isotropic Weyl dynamical maps \textnormal{(\ref{three-types-Weyl-maps})} are Markovian semigroups if and only if the probability distribution function takes the form $p(t)=\frac{|G|-1}{|G|}(1-e^{-ct})$ for some constant $c\in\mathbb{R}^+$.
\end{Proposition}

The proof of Proposition \ref{Prop1-weyl-semigroup-condition} is given in Appendix \ref{Prop:Marovian semigroup}.

The special choice of $p(t)$ ensures that all nontrivial eigenvalues evolve exponentially as $e^{-ct}$, which guarantees the semigroup property. By Proposition \ref{Prop1-weyl-semigroup-condition}, the semigroup property depends on the order of the subgroup $|G|$ through the admissible form of the mixing parameter $p(t)$.

\section{Convex combination of isotropic Weyl dynamical maps for any \texorpdfstring{$d\geq2$}{d>2}}\label{sec:4}

The isotropic Weyl dynamical maps $\mathcal{E}_{\textnormal{iso}}(t)$ (provided that all constituent maps share the same mixing parameter $p(t)$) of (\ref{three-types-Weyl-maps}) with the types (1)-(3) can be rewritten as
\begin{equation}\label{three-types-convex-combiantion}
\mathcal{E}_{\textnormal{iso}}(t)=\frac{1}{|G|-1} \sum_{u\in G\setminus\{0\}}\mathcal{E}_{u}^p(t). 
\end{equation}
That is, the isotropic Weyl dynamical maps $\mathcal{E}_{\textnormal{iso}}(t)$ can be regarded as the convex combination of $|G|-1\ (|G|\geq3)$ different Weyl dephasing maps $\mathcal{E}_{u}^p(t)$ with equal weights.

\subsection{Mixing eternally non-Markovian Weyl dephasing maps could generate Markovian semigroups}
Based on the classification of subgroups in Lemma \ref{lem:HNF}, a subgroup $G$ may possess a non-cyclic structure. Nevertheless, any such subgroup can be decomposed into a union of distinct cyclic subgroups, expressed as $G = \bigcup_{k} H_k$. Following this decomposition, for any chosen non-identity element $u:=(i,j) \in G$, we denote its associated cyclic subgroup simply by
\begin{equation}\label{H_subgroup}
    H = \langle u \rangle \subseteq G.
\end{equation}
For each such cyclic subgroup, we claim that the eternal non-Markovianity of Weyl dephasing maps associated with all elements in $H$ can be reduced to analyzing the map $\mathcal{E}_{u}^p(t)$ corresponding to the generator $u$. 

 To demonstrate this, we observe that for any non-identity generator $u=(i,j)$, the eigenvalues \eqref{eq:eigenvalue_general} of the Weyl dephasing map $\mathcal{E}_{u}^p(t)$ reduce to
\begin{equation}\lambda_{v}(t) = 1 - (1 - \omega^{u \wedge v})p(t).\end{equation}

Let $s =\gcd(u, d)= \gcd(i,j, d)$ be the greatest common divisor of $i, j$, and $d$ (where $1 \le s \le \frac{d}{2}$), then the order of the cyclic subgroup $H$ is $\ell = d/s$. Consequently, the Weyl dephasing maps associated with the remaining $\ell-2$ non-identity elements in $H$ share the exact same  eigenvalues as $\mathcal{E}_{u}^p(t)$.

Furthermore, according to \eqref{eq:gamma_def_appendix}, the decay rates $\gamma_{\alpha}(t)$ depend on the eigenvalues of the time-local generator, which take the form:
\begin{equation}
    \mu_v(t) = \frac{\dot{p}(t)(\omega^{u \wedge v} - 1)}{1 - p(t) + p(t)\omega^{u \wedge v}}.
\end{equation}
\color{black}
Notably, $\mu_v(t)$ depends on the index $v$ through the value of the symplectic inner product $u \wedge v$. Therefore, the eternal non-Markovianity of Weyl dephasing maps associated with all elements in $H$ can be reduced to analyzing the map $\mathcal{E}_{u}^p(t)$ corresponding to the generator $u$. Consequently, determining the eternal non-Markovianity of these Weyl dephasing maps amounts to analyzing the sign of the decay rates associated with the generator $u$ within the cyclic subgroup $H$. 

Based on the parity of $\ell$, we investigate the conditions under which the Weyl dephasing maps $\mathcal{E}_{u}^p(t)$ exhibit eternal non-Markovianity.

\begin{Lemma}\label{lem: ENM_Weyl_dephasing}
Suppose that the probability distribution function takes the monotonically increasing form $p(t) = r(1 - e^{-ct})$, where $r \in (0, 1]$ and $c \in \mathbb{R}^+$. Then the Weyl dephasing map $\mathcal{E}_{u}^p(t)$ is eternally non-Markovian if either $\ell \ge 3$ is odd , or $\ell \ge 2$ is even with the amplitude restricted to $0<r \le 1/2$.
\end{Lemma}
The proof of Lemma \ref{lem: ENM_Weyl_dephasing} is given in Appendix \ref{sec:AppendixC}.

\begin{Remark}\label{rem:irreducible_ENM}
\textbf{(Irreducible eternal non-Markovianity)} 

To appreciate the significance of Lemma \textnormal{\ref{lem: ENM_Weyl_dephasing}}, it is instructive to contrast it with the well-known Pauli case. To generate eternal non-Markovianity in the Pauli framework, one must construct convex combinations of two different Pauli dephasing semigroups \textnormal{\cite{Jagadish2025-vd}}, while three-way mixture can not lead to ENM. 

However, Lemma \textnormal{\ref{lem: ENM_Weyl_dephasing}} demonstrates that in higher-dimensional systems ($d \ge 3$), mixing is not a strict prerequisite for memory effects. There exist individual Weyl dephasing maps generated by a single Weyl operator that inherently exhibit ENM behavior. 
\end{Remark}

Combined Proposition \ref{Prop1-weyl-semigroup-condition} with Lemma \ref{lem: ENM_Weyl_dephasing}, we investigate whether the convex combinations of ENM Weyl dephasing maps could generate Markovian semigroups. Specifically, we analyze the structure of the constituent maps within the subgroup $G$.

\begin{Theorem}\label{Thm:1}
Consider the Markovian semigroup $\mathcal{E}_{\textnormal{iso}}(t)$ given by the convex combination \textnormal{(\ref{three-types-convex-combiantion})} of Weyl dephasing maps $\mathcal{E}_{u}^p(t)$ with the probability distribution function $p(t)=\frac{|G|-1}{|G|}(1-e^{-ct})$. For each non-identity generator $u \in G$ with $|G| \geq 3$, let $s=\gcd(u, d)$ and $\ell := d/s$. 
\begin{enumerate}
\item [\rm (1)] If $\ell \geq 3$ is odd, then the Markovian semigroup could be constructed from a convex combination of eternally non-Markovian Weyl dephasing maps.
    
\item [\rm (2)] If $\ell \geq 2$ is even, then the Markovian semigroup is generated by the Weyl dephasing maps that are not eternally non-Markovian.
\end{enumerate}
\end{Theorem}

\begin{proof}
The probability distribution function for the constituent maps is given by $p(t) = r(1-e^{-ct})$ with the specific amplitude $r = \frac{|G|-1}{|G|}$.

(1) According to Lemma \ref{lem: ENM_Weyl_dephasing}, any Weyl dephasing map $\mathcal{E}_{u}^p(t)$ is ENM provided that $\ell \ge 3$ is odd, regardless of the value of $r \in (0,1]$. Since this condition holds for every $u \in G \setminus \{0\}$, every constituent map in the convex combination (\ref{three-types-convex-combiantion}) is ENM. Furthermore, as established in Proposition \ref{Prop1-weyl-semigroup-condition}, the resulting mixture $\mathcal{E}_{\textnormal{iso}}(t)$ forms a Markovian semigroup. This demonstrates the counter-intuitive result that a Markovian semigroup could be constructed from the mixture of eternally non-Markovian dynamics.

(2) According to Lemma \ref{lem: ENM_Weyl_dephasing}, a Weyl dephasing map with even $\ell$ is ENM if the amplitude satisfies $0 < r \leq \frac{1}{2}$. We substitute the specific amplitude for the semigroup into this inequality:
\begin{equation*}
    r = \frac{|G|-1}{|G|} \leq \frac{1}{2} \implies |G| \leq 2.
\end{equation*}
However, the convex combination (\ref{three-types-convex-combiantion}) is defined for subgroups of order $|G| \geq 3$, which is a contradiction. Consequently, the constituent Weyl dephasing maps in the mixture are not ENM. Thus, the Markovian semigroup is generated by mixing maps that are not eternally non-Markovian.
\end{proof}

\begin{Remark}\label{rem:4}
The Markovian semigroup generated by the mixture of eternally non-Markovian Weyl dephasing maps is related to the parity of the ratio $\ell:=d/s$. In contrast to Ref. \cite{wudarski2016markovian}, which considered a quantum channel defined by a projector, the present work focuses on random unitary dynamical evolution, specifically the Weyl map. The first part of Theorem \ref{Thm:1} provides an additional contribution to the theory that convex combinations of eternally non-Markovian Weyl dephasing maps may generate a Markovian semigroup. 
\end{Remark}

\subsection{Convex combinations of $N$ different Weyl semigroups lead to eternal non-Markovianity}

The emergence of non-Markovianity through convex combinations of Pauli dynamical maps has been investigated in \cite{jagadish_convex_2020,Jagadish2025-vd}. As shown in \cite{Jagadish2025-vd}, an individual Pauli map cannot be ENM. ENM can arise from mixtures of two distinct Pauli dephasing semigroups, whereas mixtures involving all three Pauli dephasing semigroups fail to produce ENM. To this end, we aim to extend this discussion and explore the eternal non-Markovianity obtained by mixing high-dimensional Weyl semigroups.

In this subsection, we investigate the emergence of ENM under the convex combination of $N$ different isotropic Weyl semigroups, expressed as:
\begin{equation}\label{cc-n-Weyl-semigroups}
 \tilde{\mathcal{E}}(t)=\sum_{k=1}^N x_k\mathcal{E}_{\textnormal{iso}}^{(k)}(t). 
\end{equation}
Here, $N$ denotes the number of constituent Weyl semigroups, and $x_k$ are the mixing coefficients that satisfying $0 < x_k < 1$ and $\sum_{k=1}^N x_k = 1$. We assume that all constituent semigroups $\mathcal{E}_{\textnormal{iso}}^{(k)}(t)$, defined in (\ref{three-types-Weyl-maps}), correspond to subgroups $G_k$ of a uniform order $|G_k| = K\ (K\ge 2)$. Accordingly, the probability distribution function takes the form $p(t)=\frac{K-1}{K}(1-e^{-ct})$, where $K$ is given by \eqref{K_order}.

\color{black}

Based on the three types of the subgroup $G_k$, we analyze the eternal non-Markovianity of convex combinations (\ref{cc-n-Weyl-semigroups}).

\begin{Theorem}\label{Thm:General_NM_Overlap}
Consider the convex combination \eqref{cc-n-Weyl-semigroups} with $N$ distinct Weyl semigroups. Assume all the Weyl semigroups correspond to the subgroups $G_k$ with the same order $|G_k| = K\ (K\ge2)$, and the probability distribution function $p(t) = \frac{K-1}{K}(1-e^{-ct})$, where $K$ is given by \eqref{K_order}. Then the convex combination $\tilde{\mathcal{E}}(t)$ is eternally non-Markovian if the number of mixtures $N$ satisfies the following upper bound:
\begin{equation}\label{eq:N_bound_general}
    2 \leq N < \textnormal{min}\left\{\frac{d^2 - 1}{K - 1},\ \mathcal{N}(K) \right\},
\end{equation}
where $\mathcal{N}(K)$ is given by \eqref{eq:total_K}. 

\end{Theorem}

The proof of Theorem \ref{Thm:General_NM_Overlap} is given in Appendix \ref{sec:AppendixE}.

\begin{Remark}\label{rem:compare_with_pauli_mixture}
\textbf{(Comparison with Pauli map mixtures)} 
Theorem \textnormal{\ref{Thm:General_NM_Overlap}} provides a general criterion for the emergence of ENM from the mixture of Markovian semigroups. The significance of the upper bound in \eqref{eq:N_bound_general} becomes particularly evident when we compare it with the Pauli maps.

For qubit system, the discrete phase space $\mathbb{Z}_2 \times \mathbb{Z}_2$ has only $d^2-1 = 3$ non-identity elements, and the cyclic subgroups generated by these elements all have order $K=2$. Substituting these specific values into our bound \eqref{eq:N_bound_general} yields
\begin{equation*}
    2\le N < \textnormal{min}\left\{\frac{2^2 - 1}{2 - 1},\ \mathcal{N}(2)\right\}= 3.
\end{equation*}
Since the number of mixtures $N$ must be an integer, this inequality enforces $N=2$. This recovers the well-known Pauli result: mixing two different Pauli dephasing semigroups guarantees ENM \textnormal{\cite{Jagadish2025-vd}}. However, mixing all three Pauli dephasing semigroups covers the entire operator basis, yielding the completely depolarizing map, which is not ENM.

For high-dimensional qudit systems, the phase space capacity $d^2-1$ grows quadratically, while the subgroup coverage $K-1$ (e.g., $K=d$) typically grows only linearly. This leads to a significantly larger upper bound for $N$. For example, when $d=3$ and $K=3$, the bound becomes $2\le N < 4$, allowing the three-way mixing and still yielding ENM. Consequently, in Weyl dynamical maps, one can mix a much larger number of distinct Markovian semigroups and still preserve ENM.
\end{Remark}

\section{Examples: Weyl semigroups mixtures for the case of \texorpdfstring{$d=3$}{d=3}}\label{sec:5}

In this section, we analyze the specific case of $d=3$ in Theorem \ref{Thm:General_NM_Overlap}. The discrete phase space $\mathbb{Z}_3 \times \mathbb{Z}_3$ contains $d^2=9$ elements. According to Lemma \ref{lem:counting_general}, there are $\mathcal N(3)=\sigma_1(3)=4$ distinct subgroups of order $3$. Furthermore, any two distinct subgroups satisfy $G_i \cap G_j = \{(0,0)\}$ for $i \neq j$.

We consider the convex combination of $N$ distinct isotropic Weyl semigroups, $\tilde{\mathcal{E}}(t) = \sum_{k=1}^N x_k \mathcal{E}_{\textnormal{iso}}^{(k)}(t)$, where each Weyl semigroup $G_k$ has the same order $K=3$ and the probability distribution function $p(t) = \frac{2}{3}(1-e^{-ct})$. We focus on the nontrivial cases $2\le N\le4$.

\subsection{Eternal non-Markovianity: \texorpdfstring{$N=2$}{N=2} and \texorpdfstring{$N=3$}{N=3}}
For the prime dimension $d=3$ and the order of subgroup $|G_k|=3$, the symplectic dual subgroups satisfy $G_k^\perp=G_k$ for all $k=1,\dots,N$. Since each subgroup of order $3$ is a one-dimensional isotropic subspace of the discrete phase space, this implies that the Weyl operators within each subgroup $G_k$ are mutually commuting.

For $N=2$ and $N=3$, the number of mixtures satisfies the condition $N<4$ and the union of the subgroups $G_k$ does not cover the entire discrete phase space $\mathbb{Z}_3 \times \mathbb{Z}_3$. Therefore, there always exists at least one index $\alpha\notin \bigcup_{k=1}^N(G_k\setminus\{0\})$ such that the decay rate $\gamma_\alpha(t)$ is negative. In this case, the intersection set of \eqref{Inter_set} satisfies $S_{\text{int}}= \varnothing$, which implies $M(u,t) = 0$ in \eqref{M(u,t)}. Then the decay rate \eqref{eq:gamma_rigorous} reduces to
\begin{align}\label{d=3_reduce-decay-rate}
\gamma_\alpha(t) = \frac{1}{9} \left( c - \sum_{k=1}^N f(x_k) \right),
\end{align}
where $f(x) = \frac{c x}{x + (1-x)e^{-ct}}$. Since $x_k+(1-x_k)e^{-ct}<1 \ (t>0),$
it follows that
\[
f(x_k)
=
\frac{cx_k}{x_k+(1-x_k)e^{-ct}}
>
cx_k.
\]
Therefore,
\[
\sum_{k=1}^N f(x_k)
>
c\sum_{k=1}^N x_k
=
c,
\]
which implies $\gamma_\alpha(t)<0$ for any $t>0$. Consequently, these convex combinations are ENM.

\begin{Remark}
For the three-way mixing, there always exist at least two negative decay rates.  Moreover, the example 2 of Ref. \cite{wudarski2015} is a special case of our result. In fact, consider the mixture of  the following three isotropic Weyl semigroups:
\begin{align*}
\mathcal{E}^{(1)}_{\textnormal{iso}}(t)(\rho)&=(1-p(t))\rho+\frac{p(t)}{2}\left(U_1\rho U_1^\dagger+U_2\rho U_2^\dagger\right),  \\
\mathcal{E}^{(2)}_{\textnormal{iso}}(t)(\rho)&=(1-p(t))\rho+\frac{p(t)}{2}\left(U_3\rho U_3^\dagger+U_6\rho U_6^\dagger\right) , \\
\mathcal{E}^{(3)}_{\textnormal{iso}}(t)(\rho)&=(1-p(t))\rho+\frac{p(t)}{2}\left(U_5\rho U_5^\dagger+U_7\rho U_7^\dagger\right).  \\
\end{align*}
For notational convenience, the double index $(i,j)$ of the Weyl operators $U_{ij}$ is mapped to a single index $\alpha$ via the relation $\alpha = 3i + j$, where $i, j \in \{0, 1, 2\}$. Let the mixing coefficients $x_k=\frac{1}{3}\ (k=1,2,3)$ and $c=3c'$, then the decay rates (\ref{d=3_reduce-decay-rate}) reduce to
\begin{align*}
&\gamma_k(t)=\frac{c'}{3} \ \ (k\neq4,8), \\
&\gamma_4(t)=\gamma_8(t)=-\frac{2c'}{3}\frac{e^{2c't}-e^{c't}}{e^{2c't}+2e^{c't}}.
\end{align*}
The corresponding generators are the three isotropic Weyl semigroups $\mathcal{E}^{(k)}_{\textnormal{iso}}(t)
\ (k=1,2,3)$ is generated by
\begin{align*}
\mathcal{L}^{(1)}(t)(\rho)&=c'(U_1\rho U_1^\dagger+U_2\rho U_2^\dagger-2\rho), \\
\mathcal{L}^{(2)}(t)(\rho)&=c'(U_3\rho U_3^\dagger+U_6\rho U_6^\dagger-2\rho), \\
\mathcal{L}^{(3)}(t)(\rho)&=c'(U_5\rho U_5^\dagger+U_7\rho U_7^\dagger-2\rho),
\end{align*}
respectively. 
\end{Remark}

\subsection{Non-ENM and conditional Markovianity: \texorpdfstring{$N=4$}{N=4}}
For the four-way mixing, the union of these subgroups covers the entire discrete phase space $\mathbb{Z}_3\times \mathbb{Z}_3$. Every nonzero index $\alpha$ belongs to exactly one subgroup $G_j$. By \eqref{simplify_decay_rate}, the decay rate associated with such an index $\alpha$ is
\begin{align}\label{four-way-mixing}
\gamma_\alpha(t)= \frac{1}{9} \left[ c - \sum_{j\neq k} f(x_j) + 2f(x_k) \right].
\end{align}

In this full-coverage scenario, we demonstrate that the convex combination can exhibit either Markovian or non-Markovian behavior, which depends on the choice of the mixing coefficients $x_k$.
\begin{enumerate}
\item [\rm (1)] \textbf{Markovian regime (uniform mixing):} Consider the case of uniform mixing where $x_k = 1/4$ for all $k=1,2,3,4$. Then the decay rate (\ref{four-way-mixing}) reduces to
\begin{align}
\gamma_\alpha(t)= \frac{1}{9} \left[ c - f(1/4)\right]=\frac{c}{3(e^{ct}+3)},
\end{align}
which is positive for any $t\geq0.$ Since the decay rates depend continuously on the mixing coefficients $x_k$, positivity at the uniform point implies the existence of a neighborhood of uniform mixtures for which the dynamics remains Markovian.

\item [\rm (2)] \textbf{Non-Markovian regime (non-uniform mixing):} Consider the limit where the weight of a specific Weyl semigroup (say $\mathcal{E}_{\textnormal{iso}}^{(k)}(t)$) approaches zero: $x_k \to 0$, while keeping the summation of the rest weights approaches $1$: $\sum_{j \neq k} x_j \approx 1$. Thus, we have $\lim_{x_k \to 0} 2f(x_k) = 0$.
Denote $S_{\text{rest}} = \sum_{j \neq k} f(x_j)$ as the summation term. Since
\begin{align*}
x_j+(1-x_j)e^{-ct}&=x_j(1-e^{-ct})+1 \\
&<\sum_{j \neq k}x_j(1-e^{-ct})+1
\qquad (t>0), 
\end{align*}
one has 

\begin{align*}
 \sum_{j \neq k} f(x_j)&>\frac{c\sum_{j \neq k}x_j}{\sum_{j \neq k}x_j+(1-\sum_{j \neq k}x_j)e^{-ct}}=f(\sum_{j \neq k} x_j).
\end{align*}

Therefore, we have
\begin{align*}
S_{\text{rest}} = \sum_{j \neq k} f(x_j)> f(\sum_{j \neq k} x_j)\approx f(1)=c.
\end{align*}

Consequently, the limit of decay rate satisfies 
$$\lim_{x_k \to 0}\gamma_\alpha(t)= \frac{1}{9} \left( c - S_{\text{rest}}\right)<0$$
for any mixing coefficients $x_k$ and $t\ge 0.$

This shows that sufficiently non-uniform mixtures can produce negative decay rates and hence non-Markovian dynamics. In particular parameter regimes, the dynamics may become eternally non-Markovian.
\end{enumerate}

\subsection{Comparison: Generalized Pauli maps \emph{vs} Weyl maps}

Generalized Pauli maps, which have been extensively investigated in Refs. \cite{chruscinski2016generalized,siudzinskajpa2020}, are defined as:
\begin{equation}\label{GPC}
\mathcal{E}_{\text{GP}}(t)(\rho) = q_0(t)\rho + \frac{1}{d-1}\sum_{\alpha=1}^{d+1}q_{\alpha}(t)\mathbb{V}_\alpha(\rho). 
\end{equation}
Here, $\{q_\alpha(t)\}_{\alpha=0}^{d+1}$ constitutes a probability distribution. For each $\alpha=1,\dots, d+1$, the map acts via $\mathbb{V}_\alpha(\rho)=\sum_{k=1}^{d-1}V_{\alpha}^k\rho (V_{\alpha}^{k})^\dagger$, where the unitary operators are defined as $V_{\alpha}=\sum_{l=0}^{d-1}\omega^l|\psi_l^{(\alpha)}\rangle\langle \psi_l^{(\alpha)}|$ with $\omega=e^{2\pi\mathrm{i}/d}$. The basis sets $\mathcal{B}_\alpha=\{|\psi_l^{(\alpha)}\rangle\}_{l=0}^{d-1}$ correspond to mutually unbiased bases (MUBs). It is established that a complete set of $d+1$ MUBs exists when the dimension $d$ is a prime power \cite{wootters1989,durt2010}. Consequently, generalized Pauli maps are typically defined and discussed within the context of prime power dimensions.

We demonstrate that generalized Pauli maps can be realized as special cases of Weyl dynamical maps. To illustrate this connection explicitly, consider the case of $d=3$. The MUBs are given by
\begin{align*}
\mathcal{B}_1 &= \frac{1}{\sqrt{3}}\left\{ (1,1,1)^T, (1,\omega,\omega^2)^T, (1,\omega^2,\omega)^T \right\}, \\
\mathcal{B}_2 &= \left\{ (1,0,0)^T, (0,1,0)^T, (0,0,1)^T \right\}, \\
\mathcal{B}_3 &= \frac{1}{\sqrt{3}}\left\{ (1,\omega,\omega)^T, (1,\omega^2,1)^T, (1,1,\omega^2)^T \right\}, \\
\mathcal{B}_4 &= \frac{1}{\sqrt{3}}\left\{ (1,\omega^2,\omega^2)^T, (1,1,\omega)^T, (1,\omega,1)^T \right\},
\end{align*}
where  superscript $T$ denotes the transpose. By comparing these with the standard Weyl operators, we can identify the following relationships between the unitary operators $V_\alpha$ and the Weyl operators $U_k$:
\begin{equation*}
V_1=U_1, \quad V_2=U_3, \quad V_3=\omega^2 U_4, \quad V_4=\omega U_7.
\end{equation*}
Note that the global phase factors cancel in the unitary conjugation
$V\rho V^\dagger.$ Consequently, the generalized Pauli map \eqref{GPC} for $d=3$ can be reformulated in terms of Weyl operators as:
\begin{align}\label{d=3-GPC}
\mathcal{E}_{\text{GP}}(t)(\rho) &= q_0(t)\rho + \frac{1}{2}\bigg[ q_1(t)\left(U_1\rho U_1^\dagger + U_2\rho U_2^\dagger\right) \nonumber\\
&\quad + q_2(t)\left(U_3\rho U_3^\dagger + U_6\rho U_6^\dagger\right) \nonumber\\
&\quad + q_3(t)\left(U_4\rho U_4^\dagger + U_8\rho U_8^\dagger\right) \nonumber\\
&\quad + q_4(t)\left(U_5\rho U_5^\dagger + U_7\rho U_7^\dagger\right) \bigg].
\end{align}
Comparing this with the Weyl dynamical map defined in \eqref{double-indices-map}, we observe that \eqref{d=3-GPC} corresponds to a Weyl map with specific constraints on the coefficients: $p_0(t)=q_0(t)$, and for the non-identity terms, pairs of indices share the same weight, i.e., $p_1(t)=p_2(t)=\frac{1}{2}q_1(t)$, $p_3(t)=p_6(t)=\frac{1}{2}q_2(t)$, $p_4(t)=p_8(t)=\frac{1}{2}q_3(t), p_5(t)=p_7(t)=\frac{1}{2}q_4(t)$. More generally, in prime-power dimensions, generalized Pauli maps may be represented as Weyl maps whose weights are distributed uniformly over commuting subsets associated with mutually unbiased bases.

This implies that the analysis of convex combinations of Weyl semigroups naturally encompasses the convex combinations of generalized Pauli semigroups. For instance, consider the specific case in \eqref{d=3-GPC} where $q_0(t)=1-p(t)$, $q_1(t)=x_1 p(t)$ and $q_2(t)=x_2 p(t)$ with $x_1+x_2=1$, while $q_3(t)=q_4(t)=0$. Let the probability distribution function be $p(t)=\frac{2}{3}(1-e^{-ct})$. In this setting, the mixture leads to four negative decay rates for any mixing coefficients $0 < x_k < 1$. This result generalizes the two-way mixing scenario with equal weights ($x_1=x_2=\frac{1}{2}$) previously investigated in example 5 of Ref. \cite{siudzinskajpa2020}, demonstrating that the mixture can produce four negative decay rates over a broad range of mixing coefficients.

\section{Conclusions and discussions}\label{sec:6}
In this work, we developed a comprehensive framework for analyzing Weyl dynamical maps and the emergence of non-Markovianity in finite-dimensional open quantum systems. By combining tools from quantum dynamical semigroup theory, finite phase-space algebra, and subgroup classification over $\mathbb{Z}_d \times \mathbb{Z}_d$, we established a unified algebraic approach to the study of Weyl maps and their convex structures.

A central structural contribution of this paper is the complete classification of subgroups of the discrete phase space $\mathbb{Z}_d \times \mathbb{Z}_d$ using Hermite normal form techniques. This classification enabled us to identify and characterize three distinct classes of Weyl dynamical maps, including cyclic, split rank-$2$, and non-split rank-$2$ structures. In particular, we derived explicit counting formulas for subgroups with arbitrary order and related their algebraic properties directly to the dynamical behavior of the associated Weyl maps. These results establish a concrete bridge between finite abelian group theory and the theory of open quantum dynamics.

Building upon this framework, we investigated the semigroup structure of isotropic Weyl dynamical maps. We demonstrated that anisotropic Weyl maps with nonuniform weight distributions cannot generate semigroups in the presence of multiple distinct nontrivial eigenvalues. In contrast, we provided a characterization of the necessary and sufficient conditions under which isotropic Weyl maps form Markovian semigroups. These findings reveal that the semigroup property is  constrained by the subgroup structure underlying the Weyl maps.

One of the main findings of this work is the discovery of \emph{irreducible eternal non-Markovianity} in higher-dimensional Weyl dynamics. Specifically, we demonstrated that individual Weyl dephasing maps can themselves exhibit eternal non-Markovian behavior without requiring any convex mixing mechanism. This phenomenon has no analogue in the qubit Pauli setting, where eternal non-Markovianity can only emerge through carefully engineered mixtures of distinct semigroups \cite{Jagadish2025-vd}. Our results revealed a fundamental qualitative distinction between qubit and qudit non-Markovian dynamics, showing that higher-dimensional systems possess intrinsically richer memory structures.

Beyond this intrinsic mechanism, we discovered two complementary and nontrivial effects associated with convex mixtures of Weyl maps. First, we proved that convex combinations of eternally non-Markovian Weyl dephasing maps can generate Markovian semigroups. This demonstrated that memory effects are not additive under mixing and may instead be completely suppressed by collective interference between Weyl dynamical maps. Second, we established a general condition under which convex combinations of $N$ distinct Weyl semigroups generate eternal non-Markovianity. In contrast to the Pauli case, where only two-way mixing can sustain ENM, high-dimensional Weyl systems permit substantially larger mixtures while still preserving eternal memory effects. The emergence of ENM was shown to depend on the subgroup coverage of the discrete phase space, thereby linking non-Markovianity to algebraic structures.

We further illustrated these phenomena explicitly in the qutrit case $d=3$, where the algebra structures of the discrete phase space become particularly transparent. In this setting, we analyzed two-, three-, and four-way mixtures of Weyl semigroups and demonstrated the transition among Markovian, non-Markovian and eternally non-Markovian regimes depending on the subgroup coverage and mixing coefficients. Additionally, we showed that generalized Pauli maps constitute a special subclass of Weyl maps, thereby embedding several previously known results into the broader Weyl framework developed in this work.

Overall, our results reveal that the algebraic structure of discrete Weyl phase space plays a fundamental role in governing memory effects in open quantum systems. The interplay among subgroup structures, convexity, and divisibility uncovered here significantly extends the current theory of quantum non-Markovianity beyond the qubit regime. More broadly, the present work suggests that finite phase-space methods provide a natural and powerful tool for understanding high-dimensional quantum noise and memory effects.

Several important directions remain open for future investigations. It would be particularly interesting to extend the present analysis to time-dependent subgroup structures and noninvertible dynamics, etc. Another promising direction is the investigation of operational consequences of irreducible ENM in quantum information processing tasks, such as quantum communication, metrology, and error correction in qudit architectures. Finally, the connection among Weyl dynamical maps, finite symplectic geometry, and generalized resource theories of non-Markovianity may provide deeper structural insights into memory effects in complex quantum systems.


\appendix

\section{The proof of Lemma \ref{lem:HNF} }\label{sec:AppendixA}

\textbf{1. Lattice correspondence:}
Let $\pi: \mathbb{Z}^2 \to \mathbb{Z}_d \times \mathbb{Z}_d$ be a natural homomorphism defined by $\pi(x, y) = (x \bmod d, y \bmod d)$.
Let $G$ be an arbitrary subgroup of $\mathbb{Z}_d \times \mathbb{Z}_d$. We define the pre-image of $G$ in $\mathbb{Z}^2$ as:
\[
\mathbf{L} = \pi^{-1}(G) = \{ (x, y) \in \mathbb{Z}^2 \mid (x \bmod d, y \bmod d) \in G \}.
\]
Since $G$ is a subgroup, $\mathbf{L}$ is a sublattice of $\mathbb{Z}^2$. Furthermore, since the zero element $(0,0)\in G$, the kernel of $\pi$, which is the lattice $d\mathbb{Z} \times d\mathbb{Z}$, must be contained in $\mathbf{L}$.
Thus, we establish a one-to-one correspondence between subgroups $G \subseteq \mathbb{Z}_d \times \mathbb{Z}_d$ and full-rank sublattices $\mathbf{L} \subseteq \mathbb{Z}^2$ such that
\begin{equation*}
(d\mathbb{Z} \times d\mathbb{Z}) \subseteq \mathbf{L} \subseteq \mathbb{Z}^2.
\end{equation*}

\textbf{2. Hermite normal form (HNF):}
Since $\mathbb{Z}$ is a principal ideal domain, any sublattice $\mathbf{L}$ of $\mathbb{Z}^2$ possesses a basis of size 2. Let $\{\mathbf{b}_1, \mathbf{b}_2\}$ be a basis for $\mathbf{L}$. We can form a generator matrix with these basis vectors as rows. By applying elementary integer row operations (unimodular transformations), any integer matrix can be reduced to its \textit{Hermite normal form} \cite{newman1972}. For a $2 \times 2$ matrix, the HNF is unique and strictly upper triangular:
\[
M = \begin{pmatrix} m & w \\ 0 & n \end{pmatrix},
\]
where $m, n, w$ are integers and $0 \leq w < n$. The two row vectors of $M$, $(m, w)$ and $(0, n)$, form a basis for $\mathbf{L}$.

\textbf{3. Deriving constraints from modulo $d$:}
The condition $(d\mathbb{Z} \times d\mathbb{Z}) \subseteq \mathbf{L}$ implies that the vectors $(d, 0)$ and $(0, d)$ must be expressed as integer linear combinations of the basis vectors $(m, w)$ and $(0, n)$.

\begin{enumerate}
    \item [\rm (1)] \textit{Constraint on $n$:} The vector $(0, d)$ must be a multiple of $(0, n)$:
    \[ (0, d) = k \cdot (0, n) \implies d = kn. \]
    Therefore, $n \mid d$.

    \item [\rm (2)] \textit{Constraint on $m$:} The vector $(d, 0)$ must be a linear combination of the basis rows:
    \[ (d, 0) = c_1(m, w) + c_2(0, n) = (c_1 m, c_1 w + c_2 n). \]
    Comparing the first components: $d = c_1 m$.
    Therefore, $m \mid d$.

    \item [\rm (3)] \textit{Consistency constraint:} Comparing the second components from the equation above:
    \[ c_1 w + c_2 n = 0. \]
    Substituting $c_1 = d/m$, we get:
    \[ \frac{d}{m}w + c_2 n = 0 \implies c_2 n = -\frac{wd}{m}. \]
    For $c_2$ to be an integer, $\frac{wd}{m}$ must be divisible by $n$. Ignoring the sign (as $m,w,n$ are non-negative), we require:
    \[ n \ \bigg| \ \frac{wd}{m} \ \textnormal{or equivalently}\  wd \equiv 0 \pmod{mn}. \]
\end{enumerate}

From the above three steps, we conclude that the generators of the subgroup $G$ are simply the images of the HNF basis of $\mathbf{L}$ under the map $\pi$. Since the HNF for the lattice $\mathbf{L}$ is unique, the set of generators $\{(m, w), (0, n)\}$ defining $G$ is unique subject to the derived constraints. \qed

\section{The proof of Lemma \ref{lem:counting_general} }\label{sec:AppendixB}

\textbf{(1) Total count.} According to Lemma \ref{lem:HNF}, a subgroup of order $K$ is uniquely determined by the parameters $(m, w, n)$. For a fixed pair $(m, n) \in \mathcal{S}_K$, the parameter $w$ must satisfy $0 \le w < n$ and the consistency condition $n \mid \frac{wd}{m}$. 
This consistency condition can be rewritten as a linear congruence:
\begin{equation*}
    w \cdot \left(\frac{d}{m}\right) \equiv 0 \pmod n.
\end{equation*}
Let $g = \gcd(n, d/m)$. Dividing the congruence by $g$ yields $w \cdot \frac{d}{mg} \equiv 0 \pmod{\frac{n}{g}}$. Since $\gcd(\frac{d}{mg},\frac{n}{g})=1$, it requires that $w$ is a multiple of $\frac{n}{g}$. Within the restricted range $0 \le w < n$, the variable $w$ can take $g$ distinct values, i.e., $w = 0, \frac{n}{g}, \frac{2n}{g}, \dots, \frac{(g-1)n}{g}$. Therefore, the number of valid choices for $w$ is exactly $\gcd(n, d/m)$. Summing these choices over all permissible pairs $(m, n) \in \mathcal{S}_K$ yields the total count in Eq. \eqref{eq:total_K}.
\color{black}

\textbf{(2) Reduction to prime powers.} Since $\mathbb{Z}_d \times \mathbb{Z}_d$ is a finite abelian group, it can be decomposed into its Sylow $p$-subgroups based on the prime factorization of $d$. By the Chinese Remainder Theorem, we have the isomorphism:
\begin{equation*}
    \mathbb{Z}_d \times \mathbb{Z}_d \cong \prod_{i=1}^{s} \left( \mathbb{Z}_{p_i^{e_i}} \times \mathbb{Z}_{p_i^{e_i}} \right).
\end{equation*}
Any subgroup $G$ of $\mathbb{Z}_d \times \mathbb{Z}_d$ decomposes uniquely into a direct product $G \cong G_1 \times G_2 \times \cdots \times G_s$, where $G_i$ is a subgroup of the $p_i$-component $\mathbb{Z}_{p_i^{e_i}} \times \mathbb{Z}_{p_i^{e_i}}$.

Consequently, the number of subgroups of a given order is a multiplicative function. If $|G| = K = k_1 k_2 \cdots k_s$ (where $k_i$ is a power of $p_i$), then the total count is $\mathcal{N}(K) = \prod_{i=1}^s \mathcal{N}_{p_i}(k_i)$. To maximize the total product $\mathcal{N}(K)$, it suffices to maximize each factor $\mathcal{N}_{p_i}(k_i)$ independently.

\textbf{Proof of maximality for the case of $p$-group.} Consider the group $Q = \mathbb{Z}_{p^e} \times \mathbb{Z}_{p^e}$, which has order $p^{2e}$. We examine the number of subgroups $G_i$ of order $p^r$ for $0 \le r \le 2e$. Let this count be $\mathcal{N}(r)$.
\begin{itemize}
    \item \textbf{Duality (Symmetry):} There exists a duality between subgroups $G_i$ of order $p^r$ and subgroups $Q/G_i$ of order $p^{2e-r}$ via lattice duality. Specifically, the mapping $G_i \rightarrow Q/G_i$ establishes a bijection in the subgroup lattice. Thus, the sequence of counts is symmetric:
    \begin{equation*}
        \mathcal{N}(r) = \mathcal{N}(2e-r).
    \end{equation*}

    \item  \textbf{Unimodality:} The lattice of subgroups of a finite abelian $p$-group of type $(e, e)$ is known to be rank-unimodal. This means that the number of subgroups increases as the order approaches the ``middle'' rank and decreases afterwards. The sequence satisfies:
    \begin{equation*}
        1 = \mathcal{N}(0) < \mathcal{N}(1) < \cdots < \mathcal{N}(e) > \cdots > \mathcal{N}(2e) = 1.
    \end{equation*}
    Therefore, the maximum count for the $p$-component is attained at the middle exponent $r=e$, which corresponds to the subgroup order $|G_i| = p^e$.
\end{itemize}
Since each factor $\mathcal{N}_{p_i}(k_i)$ in the product is maximized when $k_i = p_i^{e_i}$, the total number of subgroups $\mathcal{N}(K)$ is maximized when:
\begin{equation*}
    K = \prod_{i=1}^{s} p_i^{e_i} = d.
\end{equation*}
Thus, the subgroups of order $d$ constitute the largest class of subgroups in $\mathbb{Z}_d \times \mathbb{Z}_d$.

\textbf{Explicit calculation for the maximal case ($K=d$).} 
To explicitly calculate the maximum number of subgroups, we substitute $K=d$ into our general counting formula. The constraint set $\mathcal{S}_d$ dictates $mn = d$. Consequently, for any divisor $m$ of $d$, the parameter $n$ is uniquely determined as $n = d/m$. Hence, one has 
\begin{equation*}
    \gcd\left(n, \frac{d}{m}\right) = \gcd\left(\frac{d}{m}, \frac{d}{m}\right) = \frac{d}{m}.
\end{equation*}
Therefore, the total count simplifies to a single summation over all divisors $m$ of $d$:
\begin{equation*}
    \mathcal{N}(d) = \sum_{m \mid d} \frac{d}{m}.
\end{equation*}
Since the quotient $d/m$ sweeps through the exact same set of divisors as $m$ itself when $m$ runs through all divisors of $d$, we can rewrite this sum as:
\begin{equation*}
    \mathcal{N}(d) = \sum_{k \mid d} k = \sigma_1(d),
\end{equation*}
which is the divisor sum function. Applying the multiplicative property of $\sigma_1$ to the prime factorization $d = \prod_{i=1}^s p_i^{e_i}$, we arrive at the closed-form expression:
\begin{equation*}
    \mathcal{N}(d) = \prod_{i=1}^s \frac{p_i^{e_i+1}-1}{p_i-1}.
\end{equation*} 
\qed

\section{The proof of Proposition \ref{prop:no_semigroup_anisotropic}}\label{sec:Appendix:ani-pro_map}

The proof proceeds by contradiction. We assume that the anisotropic Weyl map forms a semigroup, then we show that this leads to a contradiction regarding the functional form of $p(t)$.

By \eqref{eq:eigenvalue_general}, the eigenvalues of the anisotropic map \eqref{anisotropic proportional map} reduce to
\begin{equation*}
    \lambda_{v}(t) = 1 - p(t) + p(t) \sum_{u \in G \setminus \{0\}} w_{u} \omega^{u \wedge v}.
\end{equation*}
Let $\eta_{v} = \sum_{u \in G \setminus \{0\}} w_{u} \omega^{u \wedge v}$. We can rewrite the eigenvalues as
\begin{equation*}
    \lambda_{v}(t) = 1 - \beta_{v} p(t), \quad \text{where } \beta_{v} = 1 - \eta_{v}.
\end{equation*}
By the hypothesis, there exist at least two indices $a, b \in \mathbb{Z}_d \times \mathbb{Z}_d$ such that $\beta_a \neq \beta_b$ and $\beta_a, \beta_b \neq 0$.

If the anisotropic Weyl map forms a semigroup, then the generator $\mathcal{L}$ must be time-independent. Consequently, the eigenvalues of the generator,
\begin{equation*}
    \mu_{v}(t) = \frac{\dot{\lambda}_{v}(t)}{\lambda_{v}(t)} = \frac{-\beta_{v} \dot{p}(t)}{1 - \beta_{v} p(t)},
\end{equation*}
must be constants for all $v$. These conditions impose specific differential equations on the function $p(t)$. Let $\mu_a(t) \equiv -\Gamma_a$ and $\mu_b(t) \equiv -\Gamma_b$. For the index $a$, we have
\begin{equation*}
    \frac{\beta_a \dot{p}(t)}{1 - \beta_a p(t)} = \Gamma_a \implies p(t) = \frac{1}{\beta_a}(1 - e^{-\Gamma_a t}).
\end{equation*}
Similarly, one can obtain $p(t) = \frac{1}{\beta_b}(1 - e^{-\Gamma_b t})$ for the index $b$.

For the map \eqref{anisotropic proportional map} to be well-defined, the probability distribution function $p(t)$ must be unique. Therefore, the following identity must hold for all $t \ge 0$:
\begin{equation*}
    \frac{1}{\beta_a}(1 - e^{-\Gamma_a t}) = \frac{1}{\beta_b}(1 - e^{-\Gamma_b t}).
\end{equation*}
By comparing the Taylor series expansion of both sides at $t=0$:
\begin{itemize}
    \item First-order term ($t$): $\frac{\Gamma_a}{\beta_a} = \frac{\Gamma_b}{\beta_b}$.
    \item Second-order term ($t^2$): $\frac{\Gamma_a^2}{\beta_a} = \frac{\Gamma_b^2}{\beta_b}$.
\end{itemize}
Dividing the second-order term by the first implies $\Gamma_a = \Gamma_b$. Substituting this back into the first-order term yields $\beta_a = \beta_b$.

However, this contradicts the hypothesis that $\beta_a \neq \beta_b$. Therefore, it is impossible to satisfy the semigroup property for the anisotropic Weyl maps since the decay rates must be time-dependent functions. Consequently, the anisotropic Weyl map does not form a semigroup. \qed

\section{The proof of Proposition \ref{Prop1-weyl-semigroup-condition} }\label{Prop:Marovian semigroup}

By \eqref{eq:eigenvalue_general}, the eigenvalues of the isotropic Weyl map \eqref{three-types-Weyl-maps} are given by
\begin{align*}
\lambda_{v}(t)= 1-p(t) + \frac{p(t)}{|G|-1}\left(\sum_{u \in G} \omega^{u \wedge v}-1\right).
\end{align*}
By the definition of the symplectic dual subgroup $G^\perp$ in \eqref{symplectic_dual_subgroup}, the sum $\sum_{u \in G} \omega^{u \wedge v}$ equals $|G|$ if $v \in G^\perp$ and $0$ if $v \notin G^\perp$. Thus, the eigenvalues reduce to
\begin{equation*}
    \lambda_{v}(t) = 
    \begin{cases} 
        1, & \text{if } v \in G^\perp, \\
        1 - \frac{|G|}{|G|-1}p(t), & \text{if } v \notin G^\perp.
    \end{cases}
\end{equation*}
The multiplicities of the two eigenvalues follow from the relation $|G| \cdot |G^\perp| = d^2$: the former has multiplicity $|G^\perp| = d^2/|G|$, the latter has multiplicity $d^2(1 - 1/|G|)$.

Let $\Lambda(t) = 1 - \frac{|G|}{|G|-1}p(t)$. The eigenvalues of the generator are defined as $\mu_v(t) = \dot{\lambda}_v(t)/\lambda_v(t)$. Therefore, $\mu_{v}(t) $ equals $ 0$ for $v \in G^\perp$, and $\dot{\Lambda}(t)/\Lambda(t)$ for $v \notin G^\perp$.
By \eqref{eq:gamma_def_appendix}, the decay rate $\gamma_{\alpha}(t)$ for any $\alpha \neq 0$ reduces to 
\begin{align*}
    \gamma_{\alpha}(t) &= \frac{1}{d^2}\sum_{v \in \mathbb{Z}_d\times\mathbb{Z}_d} \omega^{-\alpha \wedge v} \mu_{v}(t) = \frac{1}{d^2} \frac{\dot{\Lambda}(t)}{\Lambda(t)} \sum_{v \notin G^\perp} \omega^{-\alpha \wedge v}  \nonumber \\
    &= \frac{1}{d^2} \frac{\dot{\Lambda}(t)}{\Lambda(t)} \left[ \sum_{v \in \mathbb{Z}_d\times\mathbb{Z}_d} \omega^{-\alpha \wedge v} - \sum_{v \in G^\perp} \omega^{-\alpha \wedge v} \right].
\end{align*}
The first sum over the entire discrete phase space is $0$ because $\alpha \neq 0$. For the second sum, by the property of symplectic duality, $\sum_{v \in G^\perp} \omega^{-\alpha \wedge v}$ equals $|G^\perp|$ if $\alpha \in (G^\perp)^\perp = G$, and $0$ if $\alpha \notin G$.
Since $|G^\perp| = d^2/|G|$, we obtain the general formula for the decay rates:
\begin{equation}\label{eq:general_gamma_pt}
    \gamma_{\alpha}(t) = 
    \begin{cases}
        -\frac{1}{|G|} \frac{\dot{\Lambda}(t)}{\Lambda(t)}, & \text{if } \alpha \in G \setminus \{0\}, \\
        0, & \text{if } \alpha \notin G.
    \end{cases}
\end{equation}

\textbf{Necessity:}
Assume the isotropic Weyl map $\mathcal{E}_{\textnormal{iso}}(t)$ is a Markovian semigroup. By definition, its generator must be time-independent and in the Lindblad form, which requires the decay rates to be non-negative constants. 
Let $\gamma_\alpha(t) = \Gamma > 0$ for all $\alpha \in G \setminus \{0\}$. From \eqref{eq:general_gamma_pt}, we have the differential equation:
\begin{equation*}
    -\frac{1}{|G|} \frac{\dot{\Lambda}(t)}{\Lambda(t)} = \Gamma \implies \frac{d}{dt} \ln \Lambda(t) = -|G|\Gamma.
\end{equation*}
Integrating this equation yields $\Lambda(t) = \Lambda(0) e^{-|G|\Gamma t}$. Since the dynamical map must satisfy $\mathcal{E}(0) = \mathbb{I}_d$, we have $p(0) = 0$, leading to $\Lambda(0) = 1$. 
Substituting $\Lambda(t) = 1 - \frac{|G|}{|G|-1}p(t)$ into the solution, one has
\begin{equation*}
    1 - \frac{|G|}{|G|-1}p(t) = e^{-|G|\Gamma t} \implies p(t) = \frac{|G|-1}{|G|} \left( 1 - e^{-|G|\Gamma t} \right).
\end{equation*}
Setting the positive constant $c = |G|\Gamma \in \mathbb{R}^+$, we obtain the exact required form $p(t) = \frac{|G|-1}{|G|}(1-e^{-ct})$.

\textbf{Sufficiency:}
Conversely, assume that the probability distribution function takes the form $p(t) = \frac{|G|-1}{|G|}(1-e^{-ct})$ with $c \in \mathbb{R}^+$. 
Then $\Lambda(t) = 1 - \frac{|G|}{|G|-1}p(t) = e^{-ct}$. Its logarithmic derivative is $\frac{\dot{\Lambda}(t)}{\Lambda(t)} = \frac{-c e^{-ct}}{e^{-ct}} = -c$.
Substituting this into our general decay rate formula \eqref{eq:general_gamma_pt}, one has 
\begin{equation*}
    \gamma_{\alpha}(t) = -\frac{1}{|G|} (-c) = \frac{c}{|G|}.
\end{equation*}
Since $c \in \mathbb{R}^+$ and $|G| \ge 2$, the decay rates $\gamma_\alpha(t)$ are positive constants for $\alpha \in G \setminus \{0\}$. This guarantees that the time-local generator is in the Lindblad form and is time-independent, meaning that the isotropic Weyl dynamical maps form Markovian semigroups. \qed

\color{black}

\section{The proof of Lemma \ref{lem: ENM_Weyl_dephasing} }\label{sec:AppendixC}

Before proving Lemma \ref{lem: ENM_Weyl_dephasing}, we need the following technical lemma:

\begin{Lemma}\label{Generating_function_Identity}
Let $d\geq2$ be any integer and $A,B,z\in\mathbb{C}$. We have the polynomial identity for $z$ where $z^n = 1$:
\begin{equation}\label{eq:identity}
    \frac{1}{A + Bz} = \frac{\sum_{m=0}^{n-1} (-B)^m  A^{n-1-m}  z^m}{A^n - (-B)^n}.
\end{equation}
\end{Lemma}

\begin{proof}
Multiplying the numerator on the right-hand side by $A+Bz$:
\begin{widetext}
\begin{align*}
    (A+Bz) \sum_{m=0}^{n-1} (-1)^m B^m A^{n-1-m} z^m &= \sum_{m=0}^{n-1} (-1)^m B^m A^{n-m} z^m + \sum_{m=0}^{n-1} (-1)^m B^{m+1} A^{n-1-m} z^{m+1} \\
    &=A^n+ \sum_{m=1}^{n-1}  (-1)^m B^m A^{n-m} z^m +\sum_{m=0}^{n-2} (-1)^m B^{m+1} A^{n-1-m} z^{m+1}+(-1)^{n-1} B^n z^n\\
    &= A^n- \sum_{m=0}^{n-2} (-1)^m B^{m+1} A^{n-m-1} z^{m+1}+\sum_{m=0}^{n-2}(-1)^m B^{m+1} A^{n-1-m} z^{m+1} + (-1)^{n-1} B^n z^n\\
    &=A^n+ (-1)^{n-1} B^n z^n=A^n- (-B)^n.
\end{align*}
\end{widetext}
Hence, we complete the proof of the identity (\ref{eq:identity}). 
\end{proof}

Now we are in the position to prove Lemma \ref{lem: ENM_Weyl_dephasing}.

\textit{The proof of Lemma }\ref{lem: ENM_Weyl_dephasing}. 
Let $s =\gcd(u, d)= \gcd(i,j, d)$ be the greatest common divisor of $i, j$ and $d$ (where $1 \le s \le \frac{d}{2}$), then the order of the cyclic subgroup $H$ is $\ell = d/s$. We claim that $\gamma_{\alpha}(t)$ in \eqref{eq:gamma_def_appendix} is non-zero only if $\alpha \in H$. In fact, the symplectic inner product $u \wedge v$ takes values from the set $\{ s x \mid x = 0, 1, \dots, \ell-1 \} \subset \mathbb{Z}_d$. We partition the phase space $\mathbb{Z}_d \times \mathbb{Z}_d$ into $\ell$ disjoint sets:
\begin{equation}
    V_x = \{ v \in \mathbb{Z}_d \times \mathbb{Z}_d \mid u \wedge v \equiv s x \pmod d \}.
\end{equation}
Obviously, the set $V_0$ is the symplectic dual subgroup $H^\perp$. We can rewrite \eqref{eq:gamma_def_appendix} as
\begin{equation}\label{eq:gamma_partitioned_appendix}
    \gamma_{\alpha}(t) = \frac{1}{d^2} \sum_{x=0}^{\ell-1} \mu_x(t) \left( \sum_{v \in V_x} \omega^{-\alpha \wedge v} \right),
\end{equation}
where the eigenvalues of the generator are
\begin{equation*}
    \mu_x(t) = \frac{\dot{p}(t)(\omega^{sx} - 1)}{1 - p(t) + p(t)\omega^{sx}}.
\end{equation*}
Let $f_u: \mathbb{Z}_d\times \mathbb{Z}_d \to \mathbb{Z}_d$ be a group homomorphism defined by $f_u(v) = u \wedge v$. Its kernel is the symplectic dual subgroup $H^\perp .$

Hence, each set $V_x$ is a coset of $H^\perp$ in $ \mathbb{Z}_d\times \mathbb{Z}_d$. Specifically, for any fixed $v_x \in V_x$ (such that $u \wedge v_x = sx$), any element $v \in V_x$ satisfies $u \wedge v = u \wedge v_x$. By linearity, $u \wedge (v - v_x) = 0$, which implies $\eta = v - v_x \in H^\perp$. Thus, every $v \in V_x$ can be expressed as $v = v_x + \eta$. 

Consequently, the sum of phase factors in \eqref{eq:gamma_partitioned_appendix} can be factored as
\begin{equation*}
    \sum_{v \in V_x} \omega^{-\alpha \wedge v} = \sum_{\eta \in H^\perp} \omega^{-\alpha \wedge (v_x + \eta)} = \omega^{-\alpha \wedge v_x} \sum_{\eta \in H^\perp} \omega^{-\alpha \wedge \eta}.
\end{equation*}
By virtue of the symplectic orthogonality, the sum $\sum_{\eta \in H^\perp} \omega^{-\alpha \wedge \eta}$ equals $|H^\perp|=\frac{d^2}{|H|}=sd$ if $\alpha \in (H^\perp)^\perp=H$ and $0$ otherwise. Thus, $\gamma_{\alpha}(t)$ is non-zero only if $\alpha \in H$.

For $\alpha \in H$, there exists an integer $y \in \{0, 1, \dots, \ell-1\}$ such that $\alpha = y u$. For any $v \in V_x$, the phase factor is
\begin{equation*}
    -\alpha \wedge v = -(y u) \wedge v = -y (u \wedge v) \equiv -ysx \pmod d.
\end{equation*}
Defining $\omega_\ell = e^{2\pi \mathrm{i}/ \ell}$ as the $\ell$-th root of unity, we have $\omega^{-ysx} = \omega_\ell^{-yx}$. Substituting this into Eq. \eqref{eq:gamma_partitioned_appendix} yields
\begin{align}\label{gamma_y}
    \gamma_y(t) &= \frac{1}{d^2} \sum_{x=0}^{\ell-1} \mu_x(t) \cdot \left( \omega_\ell^{-yx} \cdot |H^\perp| \right) \nonumber \\
    &= \frac{s}{d} \sum_{x=0}^{\ell-1} \omega_\ell^{-yx} \mu_x(t),
\end{align}
where 
$$\mu_x(t) = \frac{\dot{p}(t)(\omega_\ell^x - 1)}{1 - p(t) + p(t)\omega_\ell^x}.$$ 
Using Lemma \ref{Generating_function_Identity} with $A = 1 - p(t)$, $B = p(t)$ and $z=\omega_\ell^{x}$, we expand $\mu_x(t)$ as
\begin{align}\label{mu_GF}
    \mu_x(t) &= \frac{\dot{p}(t)}{D} \left( \sum_{m=0}^{\ell-1} C_m \omega_\ell^{x(m+1)} - \sum_{m=0}^{\ell-1} C_m \omega_\ell^{xm} \right),
\end{align}
where $D = A^\ell - (-B)^\ell$ and $C_m = (-B)^m A^{\ell-1-m}$. Substituting (\ref{mu_GF}) into (\ref{gamma_y}) and employing the orthogonality relation $\frac{s}{d} \sum_{x=0}^{\ell-1} \omega_\ell^{x(a-b)} = \delta_{ab}$, we obtain
\begin{align}\label{eq:general_gamma}
    \gamma_y(t) &= \frac{\dot{p}(t)}{D} \left( C_{y-1} - C_y \right) \nonumber\\ 
    &= \frac{\dot{p}(t) (-B)^{y-1}A^{\ell-1-y}\left( A+B \right)}{A^\ell - (-B)^\ell}\nonumber\\
    &= \frac{\dot{p}(t) (-B)^{y-1}A^{\ell-1-y}}{A^\ell - (-B)^\ell},
\end{align}
where the last equality holds as $A+B=1$.

\textbf{Case analysis by the parity of $\ell = d/s$:}

\begin{enumerate}
\item \textbf{$\ell$ is odd.} Here $(-B)^\ell = -B^\ell$, so $D = A^\ell + B^\ell > 0$. Choosing the index $y = \ell - 1$, Eq. (\ref{eq:general_gamma}) becomes
\begin{align*}            
\gamma_{\ell-1}(t)= \frac{-\dot{p}(t) B^{\ell-2}}{A^\ell + B^\ell}.
\end{align*}

Since $\dot{p}(t) \ge 0$, we have $\gamma_{\ell-1}(t)=\frac{-\dot{p}(t) p(t)^{\ell-2}}{(1-p(t))^\ell + p(t)^\ell} \le 0$ for all $t \ge 0$.
  
Thus, for odd $\ell$ and $t>0$, each Weyl dephasing map always possesses a negative decay rate $\gamma_{\ell-1}(t)$, making it eternally non-Markovian. 
    
\item \textbf{$\ell$ is even.} Here $(-B)^\ell = B^\ell$, so $D = A^\ell - B^\ell = (1 - 2p(t)) \sum_{k=0}^{\ell-1} A^{\ell-1-k} B^k$. 
Assuming $p(t) = r(1 - e^{-ct})$, the sign of $D$ depends on the term $1 - 2p(t)$.
\begin{enumerate}
\item [\rm (i)] \textbf {$0 < r \le 1/2$:} In this regime, $1 - 2p(t) \ge 0$ for all $t$, thus $D > 0$. Similar to the odd case, $\gamma_{\ell-1}(t) < 0$ for all $t > 0$.
\item [\rm (ii)] \textbf{$1/2 < r \le 1$:} $1-2p(t)$ changes sign at $t^* = \frac{1}{c}\ln\frac{2r}{2r-1}$. For $t > t^*$, $D < 0$, causing the sign of the decay rates to flip. This violates the condition of eternal non-Markovianity.
  \end{enumerate}
Consequently, for even $\ell$, the Weyl dephasing map is eternally non-Markovian if $0 < r \le 1/2$.
\end{enumerate}
\qed

\section{The proof of Theorem \ref{Thm:General_NM_Overlap}}\label{sec:AppendixE}

\textbf{(1) Spectral decomposition with overlaps.} 
As defined in \eqref{symplectic_dual_subgroup}, let $G_k^\perp$ be the symplectic dual of the $k$-th subgroup. Since the symplectic dual subgroups $G_k^\perp$ may overlap, a specific index $u$ may belong to multiple dual subgroups simultaneously. The eigenvalues of the $k$-th map $\mathcal{E}_{\textnormal{iso}}^{(k)}$ are $\lambda^{(k)}_u = 1$ if $u \in G_k^\perp$ and $e^{-ct}$ otherwise. Then the eigenvalues of the convex combination $\tilde{\mathcal{E}}(t)$ at the index $u$ are
\begin{equation}\label{eig:mixtures}
    \tilde{\lambda}_u(t) = \sum_{k=1}^N x_k \lambda^{(k)}_u(t) = \sum_{k: u \in G_k^\perp} x_k + \sum_{k: u \notin G_k^\perp} x_k e^{-ct}.
\end{equation}
Define
\begin{equation*}
    X_u = \sum_{k=1}^{N} x_k \mathbb{I}_{G_k^\perp}(u),
\end{equation*}
as the sum of the mixing coefficients of Weyl semigroups whose corresponding symplectic dual subgroups contain the index $u$. Note that $X_0 = 1$ (since $0 \in G_k^\perp$ for all $k$) and $X_u = 0$ if $u\notin G_k^\perp$. Here, the indicator function $\mathbb{I}_{G_k^\perp}(u)$ is defined as
\begin{equation*}
\mathbb{I}_{G_k^\perp}(u) =
\begin{cases}
1, & \text{if } u \in G_k^\perp, \\
0 ,& \text{if } u \notin G_k^\perp.
\end{cases}
\end{equation*}
Then we can rewrite the eigenvalue \eqref{eig:mixtures} as
\begin{equation*}
    \tilde{\lambda}_u(t) = X_u + (1-X_u)e^{-ct}.
\end{equation*}

\textbf{(2) The eigenvalues of generator.} The eigenvalues of the generator are 
\begin{align*}
\mu_u(t) &= \frac{\dot{\tilde{\lambda}}_u(t)}{\tilde{\lambda}_u(t)} \nonumber\\
&=\frac{-c(1-X_u)e^{-ct}}{X_u + (1-X_u)e^{-ct}} \nonumber\\
&=-c+D(u,t),
\end{align*}
where 
\begin{equation}\label{D_sign}
D(u,t)=
\begin{cases}
\frac{cX_u}{X_u + (1-X_u)e^{-ct}}>0, & \text{if } u \in G_k^\perp, \\
0 ,& \text{if } u \notin G_k^\perp.
\end{cases}
\end{equation}

\textbf{(3) The expression of decay rate.} Note that $\mu_0(t)=0$ for $u=0\in G_k^\perp$.
The decay rate $\gamma_\alpha(t)$ for $\alpha \neq 0$ is the inverse discrete Fourier transform of $\mu_u(t)$:
\begin{align}\label{eq:gamma_general_overlap}
    \gamma_\alpha(t) &= \frac{1}{d^2} \sum_{u \neq 0} \omega^{-\alpha \wedge u} \mu_u(t) \nonumber \\
    &=\frac{1}{d^2} \sum_{u \neq 0} \omega^{-\alpha \wedge u} [-c+D(u,t)]\nonumber \\
    &=\frac{1}{d^2}  \left[c+\sum_{u \neq 0} \omega^{-\alpha \wedge u}D(u,t)\right],
\end{align}
where the last equality holds for $\sum_{u \neq 0} \omega^{-\alpha \wedge u}=-1$.

\textbf{(4) Sign analysis of decay rate.} 
The sign of the decay rate depends on whether the index $\alpha$ falls within the union of the subgroups $G_k$.

\textit{ Case 1: $\alpha$ is uncovered in the union of $G_k$, i.e., $\alpha \notin \bigcup_{k=1}^N (G_k\setminus\{0\})$.}

Denote $D(u,t) = f(X_u)$, where the function $f(x)$ is defined as:
\begin{equation*}
    f(x) = x g(x), 
\end{equation*}
where $g(x)=\frac{c}{x + (1-x)e^{-ct}}$ is a monotonically decreasing function for any $x \in [0, 1]$. 

We claim that $D(u,t)\le \sum_{k=1}^N f(x_k)\mathbb{I}_{G_k^\perp}(u)$. Indeed, note that
\begin{align}\label{subadditive_property}
f\left(\sum_i x_i\right)=&\sum_i x_i\ g\left(\sum_i x_i\right) \nonumber\\
\leq&\sum_i x_ig\left( x_i\right)=\sum_i f\left(x_i\right)
\end{align}
for any $x_i\in[0,1].$ Therefore, we have
\begin{align*}
    D(u,t) = f(X_u) &= f\left(\sum_{k=1}^N x_k \mathbb{I}_{G_k^\perp}(u)\right) \nonumber \\
    &\le \sum_{k=1}^N f(x_k \mathbb{I}_{G_k^\perp}(u))=\sum_{k=1}^N f(x_k)\mathbb{I}_{G_k^\perp}(u).
\end{align*}
Note that the equality holds if and only if the index $u$ belongs to at most one symplectic dual subgroup $G_k^\perp$. In other words, the strict inequality occurs only when $u$ lies in the intersection of two or more distinct dual subgroups. 

To quantify this strict inequality caused by the overlapping, we define the non-negative overlap gap as
\begin{equation}\label{M(u,t)}
M(u,t) = \sum_{k=1}^N f(x_k)\mathbb{I}_{G_k^\perp}(u) - D(u,t),  
\end{equation}
which is strictly positive only $u$ belongs to the intersection set 
\begin{equation}\label{Inter_set}
S_{\text{int}} = \bigcup_{i \neq j} (G_i^\perp \cap G^\perp_j).
\end{equation}
Substituting \eqref{M(u,t)} into \eqref{eq:gamma_general_overlap}, one has
\begin{widetext}
\begin{align}\label{eq:gamma_rigorous}
    \gamma_\alpha(t) &= \frac{1}{d^2} \left[ c + \sum_{u \neq 0} \omega^{-\alpha \wedge u} \left( \sum_{k=1}^N f(x_k) \mathbb{I}_{G_k^\perp}(u) - M(u,t) \right) \right] \nonumber \\
    &= \frac{1}{d^2} \left[ c + \sum_{k=1}^N f(x_k) \left( \sum_{u \in G_k^\perp \setminus \{0\}} \omega^{-\alpha \wedge u} \right) - \sum_{u \neq 0} \omega^{-\alpha \wedge u} M(u,t) \right] \nonumber \\
    &= \frac{1}{d^2} \left[ \left( c - \sum_{k=1}^N f(x_k) \right) -\sum_{u \in S_{\text{int}}} \omega^{-\alpha \wedge u} M(u,t)\right].
\end{align}
\end{widetext}
Here, the last equality holds because $\sum_{u \in G_k^\perp \setminus \{0\}} \omega^{-\alpha \wedge u}=-1$ and $M(u,t)=0$ for any $u \notin S_{\text{int}}.$

Obviously, the first term $c - \sum_{k=1}^N f(x_k)$ in \eqref{eq:gamma_rigorous} is negative for $t > 0$. The second term depends on the phase $\omega^{-\alpha \wedge u}$. We can divide the following two cases:
\begin{itemize}
    \item [\rm (i)] If there are no overlaps, i.e., $S_{\text{int}}= \varnothing$, then the second term is equal to $0$, so $\gamma_\alpha(t) < 0$.
    \item [\rm (ii)] If there exist overlaps, we can choose an $\alpha\notin\bigcup_{k=1}^N (G_k\setminus\{0\})$ that satisfies $\omega^{-\alpha \wedge u} = 1$ for $u \in S_{\text{int}}$. Such an $\alpha$ exists provided that the union of subgroups $G_k$ do not span the entire discrete phase space. For this choice, the second term $-\sum_{u\in S_{\text{int}}} M(u,t) < 0$.
\end{itemize}
Thus, both terms in \eqref{eq:gamma_rigorous} are negative, ensuring $\gamma_\alpha(t)=0$ at $t=0$ and  $\gamma_\alpha(t) < 0$ for any $t > 0$. 

\textit{Case 2: $\alpha$ is covered, i.e., $\alpha \in \bigcup_{k=1}^N (G_k \setminus \{0\})$.}

Suppose $\alpha$ belongs to a specific subgroup $G_k$. Combined (\ref{eq:gamma_general_overlap}) with the definition of \eqref{M(u,t)}, the decay rate reduces to
\begin{align}\label{decay_case2}
    \gamma_\alpha(t) = \frac{1}{d^2} \bigg[ c &+ \sum_{j=1}^N f(x_j) \bigg( \sum_{u \in G_j^\perp \setminus \{0\}} \omega^{-\alpha \wedge u} \bigg) \nonumber\\
    &- \sum_{u \in S_{\textnormal{int}}} \omega^{-\alpha \wedge u} M(u,t) \bigg].
\end{align}
We evaluate the first summation term based on symplectic duality:
\begin{enumerate}
    \item [\rm (i)] For the subgroup $G_k$ containing $\alpha$ (i.e., $\alpha \in G_k$), we have $-\alpha \wedge  u\equiv 0 \pmod d$ for all $u \in G_k^\perp$. Thus,
    \begin{equation}\label{G_k-single}
        \sum_{u \in G_k^\perp \setminus \{0\}} \omega^{-\alpha \wedge u} = |G_k^\perp| - 1,
    \end{equation}
    where $|G_k^\perp|=\frac{d^2}{|G_k|}.$
    \item [\rm (ii)]  For other subgroups $G_j$ ($j \neq k$) where $\alpha \notin G_j$, one has
    \begin{equation}\label{G_j-term}
        \sum_{u \in G_j^\perp \setminus \{0\}} \omega^{-\alpha \wedge u} = 0 - 1 = -1.
    \end{equation}
\end{enumerate}
Substituting \eqref{G_k-single} and \eqref{G_j-term} into \eqref{decay_case2}, one has
\begin{widetext}
\begin{align}\label{simplify_decay_rate}
    \gamma_\alpha(t) &= \frac{1}{d^2} \left[ c + f(x_k)(|G_k^\perp| - 1) + \sum_{j \neq k} f(x_j)(-1) - \sum_{u \in S_{\textnormal{int}}} \omega^{-\alpha \wedge u} M(u,t)  \right] \nonumber \\
    &= \frac{1}{d^2} \left[  c - \sum_{j\neq k}^N f(x_j) + (|G_k^\perp|-1) f(x_k) - \sum_{u \in S_{\textnormal{int}}} \omega^{-\alpha \wedge u} M(u,t)  \right].
\end{align}
\end{widetext}
The sign of (\ref{simplify_decay_rate}) is related to the choices of the mixing coefficients $x_k,$ which may be positive or negative.

\textbf{(5) Counting constraint and existence of uncovered indices.}
The eternal non-Markovianity of the convex combination depends on the existence of at least one uncovered index $\alpha \notin \bigcup_{k=1}^N (G_k \setminus \{0\})$, as established in case 1. To guarantee this, the union of the mixed subgroups must constitute a proper subset of the discrete phase space $\mathbb{Z}_d \times \mathbb{Z}_d$.

Let $S$ denote the set of all nontrivial indices covered by the union of these $N$ subgroups:
\begin{equation*}
    S = \bigcup_{k=1}^N \big( G_k \setminus \{0\} \big).
\end{equation*}
According to the principle of subadditivity of cardinality, the size of a union is always less than or equal to the sum of the individual sizes. Therefore, we can bound the cardinality of $S$ by
\begin{align}\label{Eq:subadditivity_bound}
|S| &= \left| \bigcup_{k=1}^N \big( G_k \setminus \{0\} \big) \right| \leq \sum_{k=1}^N \big| G_k \setminus \{0\} \big| \nonumber\\ 
&= \sum_{k=1}^N (K - 1) = N(K - 1).
\end{align}
Note that the equality $|S| = N(K-1)$ holds if and only if the subgroups are \textit{pairwise disjoint} (intersecting solely at the identity). If there are non-trivial overlaps between the subgroups, the actual coverage is strictly reduced (i.e., $|S| < N(K-1)$), which makes the existence of an uncovered $\alpha$ even more likely.

We now impose the upper bound from Eq.~(\ref{eq:N_bound_general}). If the number of mixtures satisfies $N < \frac{d^2 - 1}{K - 1}$, multiplying both sides by $K-1$ (since $K \ge 2$) yields the strict inequality:
\begin{equation}\label{Eq:N_strict_bound}
    N(K - 1) < d^2 - 1.
\end{equation}
By combining the subadditivity bound \eqref{Eq:subadditivity_bound} and the inequality \eqref{Eq:N_strict_bound}, we establish the following inequalities
\begin{equation*}
    |S| \leq N(K - 1) < d^2 - 1.
\end{equation*}

Since the total number of nontrivial elements in the discrete phase space $\mathbb{Z}_d \times \mathbb{Z}_d$ is exactly $d^2 - 1$, the inequality $|S| < d^2 - 1$ guarantees that there exists at least one nontrivial element $\alpha \in \mathbb{Z}_d \times \mathbb{Z}_d \setminus \{0\}$ such that $\alpha \notin S$. 

Since $\alpha \notin S$, it directly follows that $\alpha \notin \bigcup_{k=1}^N (G_k \setminus \{0\})$. For this specific uncovered index $\alpha$, the decay rate is governed by the analysis in case 1, which implies that $\gamma_\alpha(t) = 0$ at $t = 0$ and $\gamma_\alpha(t) < 0$ for all $t > 0$.

Therefore, combining this with the total number of available subgroups $\mathcal{N}(K)$ from Lemma \ref{lem:counting_general}, as long as the valid range of mixtures satisfies
\begin{equation*}
    2 \le N < \min\left\{ \frac{d^2-1}{K-1},\  \mathcal{N}(K) \right\},
\end{equation*}
the convex combination \eqref{cc-n-Weyl-semigroups} is eternally non-Markovian. \qed

\bibliographystyle{apsrev}
\bibliography{Weyl.bib}
\end{document}